\documentclass[11pt,a4paper]{article}

\usepackage{cite}
\usepackage{tikz}\usetikzlibrary{matrix,decorations.pathreplacing,positioning}
\usepackage{hyperref}
\usepackage{amsmath,amsthm,amssymb,bm,relsize,comment}

\usepackage{setspace}
\usepackage{dsfont}
\usepackage{multirow}
\usepackage{array}
\usepackage{enumerate}
\newcolumntype{x}[1]{>{\centering\arraybackslash\hspace{0pt}}p{#1}}

\usepackage{wasysym}

\usepackage{caption}
\usepackage{subcaption}

\usepackage{enumitem}
\setitemize{itemsep=-1pt}
\setenumerate{itemsep=-1pt}

\usepackage{titling}
\usepackage[margin=2.6cm]{geometry}
\usepackage{pgfplots}
\pgfplotsset{width=10cm,compat=1.9}

\definecolor{myg}{RGB}{220,220,220}

\theoremstyle{definition}
\newtheorem{theorem}{Theorem}[section]

\newtheorem{proposition}[theorem]{Proposition}
\newtheorem{lemma}[theorem]{Lemma}

\newtheorem{definition}[theorem]{Definition}
\newtheorem{example}[theorem]{Example}

\usepackage{etoolbox}
\AtEndEnvironment{example}{\null\hfill$\blacktriangle$}%

\newtheorem{notation}[theorem]{Notation}
\newtheorem{remark}[theorem]{Remark}

\newcommand*{\myproofname}{Proof of the claim}

\usepackage{caption}
\usepackage{mathtools}

\newcommand{\numberset}{\mathds}

\newcommand{\Z}{\numberset{Z}}

\newcommand{\R}{\numberset{R}}

\newcommand{\F}{\numberset{F}}

\newcommand{\mC}{\mathcal{C}}

\newcommand{\mA}{\mathcal{A}}
\newcommand{\mN}{\mathcal{N}}

\newcommand{\mF}{\mathcal{F}}

\newcommand{\mD}{\mathcal{D}}
\newcommand{\mU}{\mathcal{U}}

\newcommand{\dH}{d^\textnormal{H}}
\newcommand{\drk}{d^\textnormal{rk}}

\newcommand{\st}{\, : \,}

\newcommand{\mE}{\mathcal{E}}

\newcommand{\mV}{\mathcal{V}}

\newcommand{\inn}{\textnormal{in}}
\newcommand{\out}{\textnormal{out}}

\newcommand{\CC}{\textnormal{C}}

\newcommand{\adv}{\textnormal{\textbf{A}}}

\usepackage{todonotes}

\newcommand{\bd}[1]{{\bf #1}}

\newcommand{\bfT}{\bf T}

\newcommand{\mincut}{\textnormal{min-cut}}

\newlength{\mynodespace}
\newcommand{\degin}{\partial^-}
\newcommand{\degout}{\partial^+}

\title{\textbf{Capacity of Non-Separable Networks with Restricted Adversaries}}

\usepackage{authblk}

\author[1]{Christopher Hojny}

\author[2]{Altan B. K\i l\i\c{c}\thanks{A. B. K. is supported by the Dutch Research Council through grant VI.Vidi.203.045.}}

\author[3]{Sascha Kurz}

\author[4]{Alberto Ravagnani\thanks{A. R. is supported by the Dutch Research Council through grants VI.Vidi.203.045, 
OCENW.KLEIN.539, and by the European Commission.}}
\affil[1,2,4]{
Eindhoven University of Technology, the Netherlands}
\affil[3]{University of Bayreuth, Germany}

\date{}

\begin{document}

\maketitle

\begin{abstract}
This paper investigates the problem of single-source multicasting over a communication network in the presence of restricted adversaries. When the adversary is constrained to operate only on a prescribed subset of edges, classical cut-set bounds are no longer tight, and achieving capacity
typically requires a joint design of the outer code and the inner (network) code. This stands in sharp contrast with the case of unrestricted adversaries, where capacity can be achieved by combining linear network coding with appropriate rank‑metric outer codes.
Building on the framework of \textit{network decoding}, we determine the exact one-shot capacity of one of the fundamental families of 2-level networks introduced in~\cite{beemer2023network}, and we improve the best currently known lower bounds for another such family. In addition, we introduce a new family of networks that generalizes several known examples, and derive partial capacity results that illustrate a variety of phenomena that arise specifically in the restricted-adversary setting. Finally, we investigate the concept of separability of networks with respect to both the rank metric and the Hamming metric. 
\end{abstract}

\medskip

\section{Introduction}

The field of network coding was originated in~\cite{ahlswede2000network,yeung2011network} and has received significant interest to date. In a nutshell, network coding is a communication strategy in which intermediate vertices are allowed to perform coding on the received inputs. 

In this paper, we focus on the classical single-source multicasting problem. An information source wishes to multicast packets drawn from a finite alphabet to several terminals through a directed acyclic network whose intermediate nodes may process information before forwarding it. In our setting, an omniscient adversary (fully aware of the network topology, the internal functions, and all transmitted symbols) may corrupt packets on a subset of the network edges. Adversarial error correction was studied in several works; see for instance~\cite{cai2006network,yeung2006network,yang2007refined} for classical upper bounds in network settings, as well as~\cite{matsumoto2007construction,yang2007construction} for constructions attaining them.

The case in which the adversary is restricted to act on a designated subset of edges was investigated in~\cite{ravagnani2018adversarial,beemer2023network}. The key insight of~\cite{beemer2023network} is that techniques valid for unrestricted adversaries fail when the adversary may act only on a \textit{proper} subset of the edges. In particular, classical cut-set bounds are not tight anymore, and new methods to design codes are required. 
The main proposal of~\cite{beemer2023network} 
is the use of \textit{network decoding}, in which intermediate nodes locally decode information before forwarding it. This departs significantly from the unrestricted setting, where capacity can be achieved using random linear network coding provided that the alphabet is large compared to the number of terminals; see~\cite{li2003linear,koetter2003algebraic,jaggi2007resilient,ho2006random,koetter2008coding,jaggi2005polynomial}.

Motivated by~\cite{beemer2023network}, the papers \cite{kurz2022capacity,hojny2023role} used optimization techniques to study specific networks for small alphabets. Optimization-based approaches have appeared earlier in network coding in different contexts, e.g.~\cite{lun2005achieving,lun2006minimum}. To date, solving the multicasting problem for small alphabets remains a major open challenge; see~\cite{chekuri2006average,guang2018alphabet,sun2016base} for recent progress, and~\cite{lehman2004complexity} for its computational complexity.

The purpose of this paper is to advance the understanding of multicasting with restricted adversaries, with emphasis on whether the design of the outer code can be decoupled from the design of the inner (network) code. We begin by determining the capacities of certain families of networks previously introduced in~\cite{beemer2023network}. In particular, we establish the capacity of Family~E and we improve the best known bounds for Family~B, extending some of the results of~\cite{kurz2022capacity}. We then introduce a new family of simple 2-level networks that generalizes Family~B. For this new family, we derive tight or nearly tight capacity results for a wide range of parameters, highlighting several structural phenomena that arise uniquely in the restricted-adversary regime.

In the final part of the paper, we study the notion of separability: Whether it is possible to design an inner network code that works universally with any outer code correcting a prescribed number of errors. We introduce a formal definition of separability and show that networks with restricted adversaries may fail to be separable with respect to the rank metric, and that even in the unrestricted case separability need not hold with respect to the Hamming metric.

The rest of the paper is organized as follows. Section~\ref{sec:intro} introduces the formal model and notation. Section~\ref{sec:threefam} studies Families~B and~E: We refine previously known bounds for the former and compute the exact capacity of the latter. Section~\ref{sec:newfamily} introduces a generalized family of networks and determines its capacity for various parameters. Section~\ref{sec:separable} develops the theory of separability and analyzes it with respect to the rank and Hamming metrics. Section~\ref{sec:conclusion} concludes with two research directions.

\section{Problem Statement and Communication Model}
\label{sec:intro}

We focus on single-source networks whose inputs come from a finite alphabet and the source wishes to send information packets to the terminals of the network through its intermediate nodes. These nodes can process the incoming information before transmitting it to the next nodes. We work with delay-free networks in which the vertices are memoryless and are interested in solving the \textit{multicast problem},
where each terminal of the network demands all the packets the source sends. 

In this paper, $q$ will always denote a prime power and we assume that the reader is familiar with elementary concepts from network coding, graph theory, and mixed-integer programming;
see \cite{b12}, \cite{b22} and \cite{b29}, respectively, among many others. We start by formally defining communication networks.

\begin{definition}
\label{def:network}
A (\textbf{single-source}) \textbf{network}  is a 4-tuple $\mN=(\mV,\mE, S, \bfT)$ 
that satisfies the following properties:
\begin{enumerate}
\item $(\mV,\mE)$ is a finite, directed, acyclic multigraph;
\item $S \in \mV$ is the \textbf{source} and $\bfT \subseteq \mV$ is the set of \textbf{terminals};
\item $\lvert \bfT \rvert \ge 1$ and $S \notin \bfT$;
\item the source does not have incoming edges, and terminals do not have outgoing edges;
\item \label{prnE} for any~$T \in \bfT$, there exists a directed path from~$S$ to~$T$;
\item for every~$V \in \mV \setminus (\{S\} \cup \bfT)$, there exists a directed path from~$S$ to~$V$ and from~$V$ to some~$T \in \bfT$. 
\end{enumerate}
The elements of~$\mV$ are called \textbf{vertices}  (or \textbf{nodes}), the elements of~$\mV \setminus (\{S\} \cup \bfT)$ are called \textbf{intermediate} vertices, and the elements of~$\mE$ are called \textbf{edges}.
\end{definition}

The other conditions in our communication model are as follows: The edges of a network carry precisely one symbol from an \textbf{alphabet} $\mA$
(a finite set with cardinality of at least 2). 
We model errors as being introduced by an \textbf{adversary} $\adv$ (adversarial model), who can act
on up to $t$ edges from a set $\mU$ of 
vulnerable edges. 
The adversary can act by changing the alphabet symbol carried by the edge into any other alphabet symbol of choice. We call $t$ the \textbf{adversarial power}. If~$t=0$ or $\mU=\emptyset$, then there is no adversary acting on the network. 
In our figures, the 
vulnerable edges will be drawn with dashes; see, for example, Figure \ref{fig:example}.

\begin{remark}
Restricting the adversary to a prescribed subset of edges is the main feature of the model considered in this work. This constraint is precisely what makes the computation of the one‑shot capacity substantially more delicate than in the unrestricted case, where classical cut‑set bounds are often tight and standard decoding strategies apply.
\end{remark}

We now establish the notation for the rest of the paper.

\begin{notation} 
\label{not:fixN}
Throughout the paper, $\mN=(\mV,\mE,S,{\bf T})$ denotes a network as defined in Definition~\ref{def:network}. We let $$\mincut_\mN(V,V')$$
be the minimum size of an edge-cut between the vertices $V,V' \in \mV$, and 
$$\mu_{\mN}= \min\{\mincut_\mN(S,T) \mid T \in \bf T\}.$$ We denote the
set of incoming and outgoing edges of an intermediate node $V \in \mV$ by $\inn(V)$ and~$\out(V)$, respectively, and their cardinalities by~$\degin(V)$ and~$\degout(V)$. 
\end{notation}

The edges of a network
can be partially ordered as follows.

\begin{definition}
\label{def:order}
Let $e,e'$ be two different edges in $\mN$. We say that $e$ \textbf{precedes} $e'$ if
there exists a directed path in $\mN$ that starts with $e$ and ends with~$e'$. In symbol, $e \preceq e'$.
\end{definition}

Since the partial order $\preceq$ can be extended to a total order $\le$, we fix such an order extension and display it by labeling the edges as $\{e_1, \ldots, e_{|\mE|}\}$, with the property that if $e_i \preceq e_j$, then~$i \le j$. We can therefore talk about the $i$th edge of the network.
The results in this work do not depend on the chosen total order.

The network is ``coded" by the functions assigned to the intermediate nodes. These functions specify how the information is processed throughout the network.

\begin{definition}
\label{def:nc}
A \textbf{network code} for~$(\mN,\mA)$ is a set of functions $$\mF=\{\mF_V \st V \in \mV \setminus (\{S\} \cup \bfT)\},$$ where $\mF_V \colon \mA^{\degin(V)} \to \mA^{\degout(V)} \quad \mbox{for all } V \in \mV \setminus (\{S\} \cup \bfT).$ If $\mA = \F_q$ is a finite field with $q$ elements and all functions in $\mF$ are $\F_q$-linear, then $\mF$ is called a \textbf{linear} network code.

\end{definition}

A network code uniquely determines how the incoming packets are processed by the intermediate nodes, since there is a total order defined on the edges; see \cite{ravagnani2018adversarial} for a detailed discussion. 
The transfer from the source to each of the terminals can be described by adversarial channels, in the sense of~\cite{beemer2023network,ravagnani2018adversarial}.

\begin{notation} 
\label{not:netch}
Let $\mN=(\mV,\mE, S, \bfT)$ be a network,~$T \in \bfT$ a terminal, $\mF$ a network code for~$(\mN,\mA)$, and $\mU \subseteq \mE$ an edge set. We denote by
$$\Omega[\mN, \mA, \mF, S \to T,\mU,t] : \mA^{\degout(S)} \rightarrow \mA^{\degin(T)}$$
the channel representing the transfer from $S$ to terminal~$T \in \bd{T}$ under the rules of the network code $\mF$ and at most $t$ symbols over the edges in $\mU$ can be attacked by the adversary. 
\end{notation}

Informally speaking, $\Omega[\mN, \mA, \mF, S \to T,\mU,t](x)$ for $x \in \mA^{\degout(S)}$ is the set of all possible words that can be seen on $\inn(T)$. When there is no 
adversary, this set is a singleton and there is no ambiguity in the decoding.
This is not the case in general;
see Example~\ref{ex:intro} below for an illustration.

We are interested in finding the maximum number of alphabet symbols that can be multicasted error-free in a single use of the network; see e.g. \cite{cai2006network, yeung2006network}. To define this parameter, we need to explain the concept of sending information without errors.

\begin{definition}
\label{def:outer}
Let~$\mN=(\mV,\mE,S,{\bf T})$ and $\mF$ be a network code for~$(\mN,\mA)$. An (\textbf{outer}) \textbf{code} for a network~$\mN=(\mV,\mE,S,{\bf T})$ is a nonempty subset~$\mC \subseteq \mA^{\degout(S)}$.
We say that~$\mC$ is \textbf{unambiguous} for the channel $\Omega[\mN,\mA,\mF,S \to T,\mU,t]$ for $T \in \bfT$ if for all~$\mathbf{x}, \mathbf{x}' \in \mC$ with~$\mathbf{x} \neq \mathbf{x}'$ we have
\[
    \Omega[\mN,\mA,\mF,S \to T,\mU,t](\mathbf{x}) \, \cap  \Omega[\mN,\mA,\mF,S \to T,\mU,t](\mathbf{x}') = \emptyset,
\] where $\adv$ acts on $\mU$ with adversarial power $t$. The elements of the outer code $\mC$
are called \textbf{codewords}.
\end{definition}

The key parameter of the network we study in this paper is the \textit{one-shot capacity}.

\begin{definition} 
\label{def:capacities}
Let $\mN=(\mV,\mE, S, \bfT)$ be a network, and
$\mU \subseteq \mE$ an edge set.
The (\textbf{one-shot}) \textbf{capacity}
of $(\mN,\mA,\mU,t)$, denoted by $$\CC_1(\mN,\mA,\mU,t)$$
is the largest 
real number~$\gamma$ such that there is an unambiguous outer code $\mC \subseteq \mA^{\degout(S)}$ for each channel
$\Omega[\mN,\mA,\mF,S \to T,\mU,t]$ with $T \in \bd{T}$
and a network code $\mF$ for~$(\mN,\mA)$ 
with $$\gamma=\log_{|\mA|}|\mC|.$$ 
\end{definition}

The only capacity considered in this paper is the one-shot capacity. Therefore, from now on, we omit the term ``one-shot''.
We continue with the state-of-the-art bound for the one-shot capacity of a network.

\begin{theorem}[Generalized Network Singleton Bound; see~\cite{ravagnani2018adversarial}] 
\label{thm:sbound}
Let $\mN=(\mV,\mE, S, \bfT)$ be a network, $\mA$ an alphabet,~$\mU \subseteq \mE$,
and $t \ge 0$ an integer. 
We have
\[
\CC_1(\mN,\mA,\mU,t) \le \min\limits_{T \in \bfT} \, \min\limits_{\mE'} \left(
\lvert \mE'\setminus \mU \rvert + \max\{0,\lvert \mE' \cap \mU \rvert-2t\} \right)
\]
where $\mE' \subseteq \mE$ ranges over edge-cuts between $S$ and $T$. In particular,
if $\mU = \mE$ then $\CC_1(\mN,\mA,\mE,t) \le \max\{0,\mu_{\mN}-2t\}$.
\end{theorem}

We give the following example to further explain the concepts introduced so far.

\begin{example}
\label{ex:intro}
Let $\mN$ be the network depicted in Figure \ref{fig:example}. The fixed total order is illustrated by the labeling of the edges. Recall that that the vulnerable edges~$\mU$ are indicated by dashed edges in the figure. 

One can check by hand that removing one edge is not enough to cut the connection between the source and any of the terminals. Thus $\mu_{\mN}=2$. 
Theorem~\ref{thm:sbound}
gives $\CC_1(\mN,\mA,\mE,1)=0$. The capacity may sometimes not increase by 
restricting the adversary. For instance, let $\adv$ be an adversary that acts on $\mU=\{e_1,e_4,e_7,e_8\}$ with adversarial power 1.  Let $\mA = \F_3$ and~$\mF = \{\mF_{V_1},\mF_{V_2},\mF_{V_3},\mF_{V_4}\}$ with 
$$\mF_{V_1}(a) = (a,a),\ \mF_{V_2}(a) = (2a,a),\ \mF_{V_3}(a,b) = a+b,\ \mF_{V_4}(a) = (a,2a)$$ for $a,b \in \F_3$. Suppose $t=1$, that is, the adversary can change a symbol on one of the edges in~$\mU$. Theorem \ref{thm:sbound} gives $\CC_1(\mN,\mA,\mU,1)=0$ (for example, take $T_1$ and $\mE'=\{e_1,e_4\}$). 

\end{example}

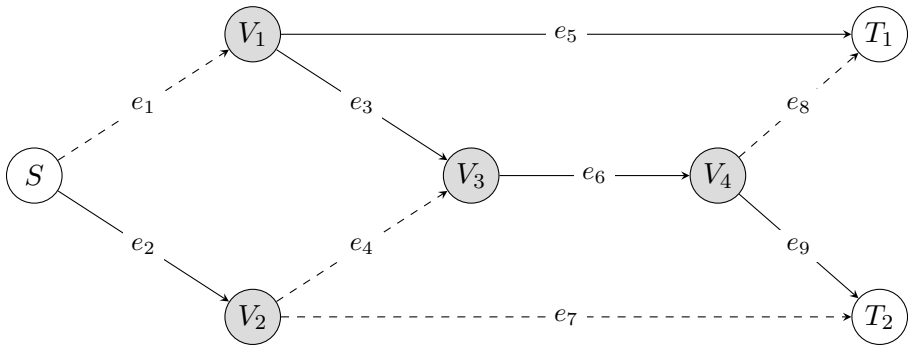
\begin{figure}[h!]
\centering
\begin{tikzpicture}
\tikzset{vertex/.style = {shape=circle,draw,inner sep=0pt,minimum size=1.9em}}
\tikzset{nnode/.style = {shape=circle,fill=myg,draw,inner sep=0pt,minimum
size=1.9em}}
\tikzset{edge/.style = {->,> = stealth}}
\tikzset{dedge/.style = {densely dotted,->,> = stealth}}
\tikzset{ddedge/.style = {dashed,->,> = stealth}}

\node[vertex] (S1) {$S$};

\node[shape=coordinate,right=\mynodespace of S1] (K) {};

\node[nnode,above=0.6\mynodespace of K] (V1) {$V_1$};

\node[nnode,below=0.6\mynodespace of K] (V2) {$V_2$};

\node[nnode,right=\mynodespace of K] (V3) {$V_3$};

\node[nnode,right=\mynodespace of V3] (V4) {$V_4$};

\node[vertex,right=3\mynodespace of V1 ] (T1) {$T_1$};

\node[vertex,right=3\mynodespace of V2] (T2) {$T_2$};

\draw[ddedge,bend left=0] (S1)  to node[fill=white, inner sep=3pt]{\small $e_1$} (V1);

\draw[edge,bend left=0] (S1) to  node[fill=white, inner sep=3pt]{\small $e_2$} (V2);

\draw[edge,bend left=0] (V1) to  node[fill=white, inner sep=3pt]{\small $e_3$} (V3);

\draw[ddedge,bend left=0] (V4)  to node[fill=white, inner sep=3pt]{\small $e_{8}$} (T1);

\draw[edge,bend left=0] (V4)  to node[fill=white, inner sep=3pt]{\small $e_{9}$} (T2);

\draw[edge,bend left=0] (V1)  to node[fill=white, inner sep=3pt]{\small $e_{5}$} (T1);

\draw[ddedge,bend left=0] (V2)  to node[fill=white, inner sep=3pt]{\small $e_{7}$} (T2);

\draw[ddedge,bend left=0] (V2) to  node[fill=white, inner sep=3pt]{\small $e_4$} (V3);

\draw[edge,bend left=0] (V3) to  node[fill=white, inner sep=3pt]{\small $e_{6}$} (V4);

\end{tikzpicture}
\caption{An example of a network.\label{fig:example}}
\end{figure}

Computing the capacity of networks becomes very difficult as they get larger and more complex. In \cite{beemer2023network}, it was shown how the capacity of an arbitrary network $\mN$ is upper bounded by the capacity of a simple 2-level network 
that can be constructed from it.

\begin{definition}
\label{def:2-level}
Let $\mN=(\mV,\mE,S,\mathbf{T})$ be a network. We say that $\mN$ is a  \textbf{simple 2-level} network if $\lvert \mathbf{T} \rvert=1$ and any path from $S$ to the terminal is of length $2$.
\end{definition}

The same paper introduces five families of simple 2-level networks (labeled from A to E) as the building blocks of the theory. The capacities of Families C and D 
have been computed in \cite{beemer2023network} as well. In this paper, we focus on Families B and E.

\section{Two Families of Simple 2-Level Networks}
\label{sec:threefam}

 In this section, we give some computational results on the networks of Family B using methods from \textit{constraint satisfaction}, and
 fully compute the capacity of the networks of Family E. All networks in these families
 have two intermediate nodes~$V_1$ and $V_2$, and the adversary can only act on the outgoing edges of the source~$S$, which we denote by $\mU_S$. 
 Any such network will be denoted by
$$([\degin(V_1),\degin(V_2)],[\degout(V_1),\degout(V_2)]).$$

It was proven in \cite[Lemma 6.5]{beemer2023network} that, 
whenever the condition $\degin(V_1) \le \degout(V_1)$ holds for a simple 2-level network, one can assume $\degin(V_1) = \degout(V_1)$ and that $\mF_{V_1}$ is the identity function.

Before analyzing the two families,
we recall the definition of the
\textit{Diamond Network}, which
represents the smallest example where the cut-set bound of Theorem \ref{thm:sbound} is not sharp in the case $\mU \neq \mE$, and one needs a non-standard technique to achieve capacity.

\begin{theorem}[Diamond Network; see~\cite{beemer2022network}]
\label{thm:diamond}
Let $\mD=([1,2],[1,1])$ and $|\mA|=q$. We have $$\CC_1(\mD,\mA,\mU_S,1) = \log_q(q-1).$$
\end{theorem}

\subsection{Family B}

The networks of Family B are parametrized by a positive integer $s$ and they are defined as follows:
$$\mathfrak{B}_s = ([1,s+1],[1,s]);$$
see Figure \ref{fig:familyb} for an illustration. Since $s=1$ gives the Diamond Network, we assume for the rest of this subsection that $s \ge 2$.
As for upper bounds,
Theorem \ref{thm:sbound} gives $\CC_1(\mathfrak{B}_s,\mA,\mU_S,1) \le s$.
However, this bound is not achievable, and in \cite{beemer2023network} it was shown that~$\CC_1(\mathfrak{B}_s,\mA,\mU_S,1) < s$. 

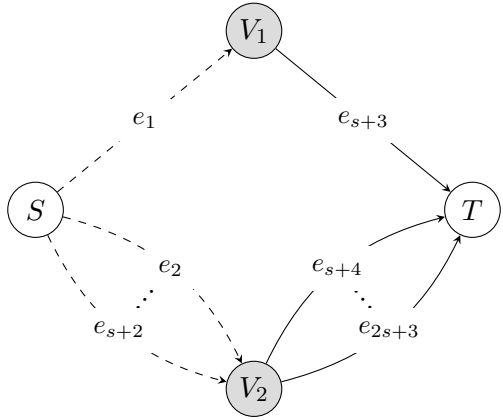
\begin{figure}[htbp]
\centering

\begin{tikzpicture}
\tikzset{vertex/.style = {shape=circle,draw,inner sep=0pt,minimum size=1.9em}}
\tikzset{nnode/.style = {shape=circle,fill=myg,draw,inner sep=0pt,minimum
size=1.9em}}
\tikzset{edge/.style = {->,> = stealth}}
\tikzset{dedge/.style = {densely dotted,->,> = stealth}}
\tikzset{ddedge/.style = {dashed,->,> = stealth}}

\node[vertex] (S1) {$S$};

\node[shape=coordinate,right=\mynodespace of S1] (K) {};

\node[nnode,above=0.8\mynodespace of K] (V1) {$V_1$};
\node[nnode,below=0.8\mynodespace of K] (V2) {$V_2$};

\node[vertex,right=2\mynodespace of S1] (T) {$T$};

\draw[ddedge,bend left=0] (S1) to  node[fill=white, inner sep=5pt]{\small $e_1$} (V1);

\draw[draw= none,bend left=5] (S1)  to node[sloped,fill=white, inner sep=5pt]{$\vdots$} (V2);
\draw[ddedge,bend left=25] (S1) to  node[fill=white, inner sep=5pt]{\small $e_{2}$} (V2);
\draw[ddedge,bend right=25] (S1) to  node[fill=white, inner sep=5pt]{\small $e_{s+2}$} (V2);

\draw[edge,bend left=0] (V1) to  node[fill=white, inner sep=5pt]{\small $e_{s+3}$} (T);

\draw[draw= none,bend left=5] (V2)  to node[sloped,fill=white, inner sep=5pt]{$\vdots$} (T);
\draw[edge,bend left=25] (V2) to  node[fill=white, inner sep=5pt]{\small $e_{s+4}$} (T);
\draw[edge,bend right=25] (V2) to  node[fill=white, inner sep=5pt]{\small $e_{2s+3}$} (T);

\end{tikzpicture} 

\caption{Family B}\label{fig:familyb}
\end{figure}

Regarding lower bounds,
when $|\mA| \ge s$ we have
$\CC_1(\mathfrak{B}_s,\mA,\mU_S,1) \ge s-1$; see~\cite{beemer2023network}. Moreover, for some specific values of $(s,|\mA|)$, the exact value of $\CC_1(\mathfrak{B}_s,\mA,\mU_S,1)$ is calculated in \cite{kurz2022capacity}. For example, when $|\mA|=2$, the number $\CC_1(\mathfrak{B}_s,\mA,\mU_S,1)$ is computed for all $s\in \{2,3,\ldots,12\}.$ There are still many open cases, some of which are listed in \cite[Table~4]{kurz2022capacity} with the best known bounds.
\begin{lemma}[see \cite{kurz2022capacity}]
\label{lem:kurz22}
The following hold.
\begin{enumerate}
    \item If $\lvert \mA \rvert =3$, then $\CC_1(\mathfrak{B}_4,\mA,\mU_S,1) \in \{\log_3(x) \mid 35 \le x \le 38\}.$
    \item If $\lvert \mA \rvert = 4$, then $\CC_1(\mathfrak{B}_2,\mA,\mU_S,1) \in \{\log_4(x) \mid 9 \le x \le 11\}$.
    \item If $\lvert \mA \rvert=4$, then $\CC_1(\mathfrak{B}_3,\mA,\mU_S,1) \in \{\log_4(x) \mid 31 \le x \le 37\}$. 
    \item If $\lvert \mA \rvert=5$, then $\CC_1(\mathfrak{B}_2,\mA,\mU_S,1) \in \{\log_5(x) \mid 15 \le x \le 17\}.$
\end{enumerate}
\end{lemma}

Our first result improves
two lower bounds of Lemma \ref{lem:kurz22}.
For $s=2$ and $|\mA|=4$, we were able to find an unambiguous code of cardinality 10 
by modeling the code search as a satisfiability problem; see Appendix~\ref{sec:MINLPmodel} for the details.
To model and solve this satisfiability problem, we used the Python module \texttt{pysat} with the \texttt{Glucose4} solver~\cite{glucose}.\footnote{The code used in our experiments is available at~\cite{zenodo} and \href{https://github.com/christopherhojny/supplement_network_codes}{GitHub}.}
Similarly, we constructed an unambiguous code of cardinality 16 for $s=2$ and $|\mA|=5$.
In Appendix~\ref{sec:concretecodes} we include examples of codes with the mentioned cardinalities.

\begin{proposition}
\label{prop:familyB}
We have 
\begin{enumerate}
    \item $\CC_1(\mathfrak{B}_2,\mA,\mU_S,1) \ge \log_4(10)$ if $\lvert \mA \rvert=4$, 
    \item $\CC_1(\mathfrak{B}_2,\mA,\mU_S,1) \ge \log_5(16)$ if $\lvert \mA \rvert=5$. 
\end{enumerate}
\end{proposition}

\subsection{Family E}
The networks of Family~E are parametrized by a 
positive integer $t$. For $t \in \Z_{\ge 1}$,
define $$\mathfrak{E}_t = ([t,t+1],[1,1]);$$ see Figure~\ref{fig:familye} for an illustration.
Note that
in Family E almost half of the edges can be attacked,
while in Family B only one of the edges is vulnerable.

Theorem \ref{thm:sbound} reads $\CC_1(\mathfrak{E}_t,\mA,\mU_S,t) \le 1$. In \cite{beemer2023network} it is shown that this inequality is strict, that is, $\CC_1(\mathfrak{A}_t,\mA,\mU_S,t) < 1$. We compute the exact capacity in Theorem \ref{thm:familyE} below. The proof relies on four lemmas that we prove separately, two for even adversarial power and two for odd adversarial power.

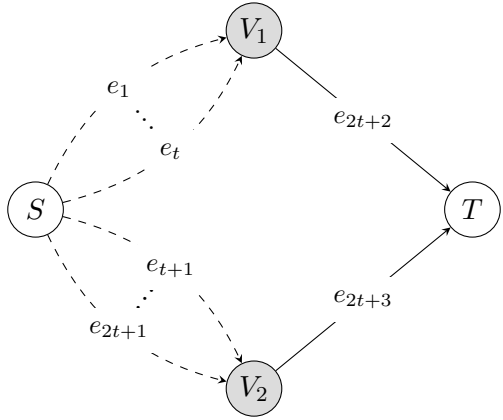
\begin{figure}[!h]
\centering
\begin{tikzpicture}
\tikzset{vertex/.style = {shape=circle,draw,inner sep=0pt,minimum size=1.9em}}
\tikzset{nnode/.style = {shape=circle,fill=myg,draw,inner sep=0pt,minimum
size=1.9em}}
\tikzset{edge/.style = {->,> = stealth}}
\tikzset{dedge/.style = {densely dotted,->,> = stealth}}
\tikzset{ddedge/.style = {dashed,->,> = stealth}}

\node[vertex] (S1) {$S$};

\node[shape=coordinate,right=\mynodespace of S1] (K) {};

\node[nnode,above=0.8\mynodespace of K] (V1) {$V_1$};
\node[nnode,below=0.8\mynodespace of K] (V2) {$V_2$};

\node[vertex,right=2\mynodespace of S1] (T) {$T$};

\draw[draw= none,bend left=5] (S1)  to node[sloped,fill=white, inner sep=3pt]{\small $\vdots$} (V1);
\draw[ddedge,bend left=25] (S1) to  node[fill=white, inner sep=5pt]{\small $e_1$} (V1);
\draw[ddedge,bend right=25] (S1) to  node[fill=white, inner sep=5pt]{\small $e_t$} (V1);

\draw[draw= none,bend left=5] (S1)  to node[sloped,fill=white, inner sep=5pt]{$\vdots$} (V2);
\draw[ddedge,bend left=25] (S1) to  node[fill=white, inner sep=5pt]{\small $e_{t+1}$} (V2);
\draw[ddedge,bend right=25] (S1) to  node[fill=white, inner sep=5pt]{\small $e_{2t+1}$} (V2);

\draw[edge,bend right=0] (V1) to  node[fill=white, inner sep=5pt]{\small $e_{2t+2}$} (T);

\draw[edge,bend right=0] (V2) to  node[fill=white, inner sep=5pt]{\small $e_{2t+3}$} (T);
\end{tikzpicture} 

\caption{Family E}\label{fig:familye}
\end{figure}

\begin{lemma}
\label{lemma-1}
Let $m$ be a positive integer, $t=2m$, and 
$a'=\left\lfloor |A|/(m+1) \right\rfloor$.
If $C$ is unambiguous for $\Omega[\mathfrak{E}_t,\mA,\mF,S \to T,\mU_S,t]$, then $|C| \le a'$.
\end{lemma}
\begin{proof}
Without loss of generality, an unambiguous code $\mC$ of the largest size can be assumed to be a repetition code. Let $\mathcal{A}'\subseteq\mathcal{A}$ be an arbitrary subset of cardinality $a'$. Our strategy is to lower bound the necessary output states of $V_2$, where $V_2$ has output state $(c,r)$ if $V_2$ receives an input that contains the symbol $c \in \mA'$ exactly~$r$ times. 

Let $a$ and $b$ be two arbitrary different elements of $\mathcal{A}'$. Consider a vector $c_2$ of length $2m+1$ such that $x$ of its coordinates are equal to $a$, and the remaining $(2m+1-x)$ coordinates are equal to $b$. Similarly, consider a vector $c_2'$ of length $2m+1$ such that $2m+1-y$ of its coordinates are equal to $a$, and the remaining $y$ coordinates are equal to $b$. Consider both vectors as input for $V_2$. 

Let $c_1$ be a vector of length $2m$ such that $z$ of its coordinates are equal to $a$ and the remaining $2m-z$ coordinates are equal to $b$. Consider $c_1$ as an input for 
  $V_1$. If 
  \begin{equation*}
    x+z\ge 2m+1 \quad\text{and}\quad y+2m-z \ge 2m+1, 
  \end{equation*}
  then the adversary can modify the codewords $(a, \dots, a)$ and $(b,\dots, b)$ to $c_1\mid c_2$ and $c_1 \mid c_2'$. Thus, $V_2$ cannot map $c_2$ and $c_2'$ to the same output states.
Next, we will show that $V_2$ must have different output states $(a,r)$ and $(b,r)$ for all $$\left\lceil\tfrac{t}{2}\right\rceil+1\le r\le 2m+1.$$ Choosing $x=i$, $y=j$, and $z=y-1$ we can distinguish output states $(a,i)$ and $(b,j)$ for all~$m+1\le i,j\le 2m+1$, since $x+z=i+j-1\ge 2m+1$. Choosing $x=i$, $y=t+1-j$, and~$z=y-1$ we can distinguish the output states $(a,i)$ and $(a,j)$ for all $m+1\le j<i\le t+1$, since $x+z=i+t+1-j-1\ge 2m+1$. Using the fact that the symbols $a$ and $b$ are arbitrary in $\mathcal{A}'$, 
  we conclude that $V_2$ has at least $\lvert \mathcal{A}' \rvert \cdot((2m+1)-(m+1)+1)) =\lvert \mathcal{A}'\rvert \cdot(m+1)$ different output states if $C$ is unambiguous. Therefore, we must have 
$$\lvert \mathcal{A}'\rvert \cdot(m+1) \le |\mA|^{\degout(V_2)}=|\mA|.$$ Since $\mC$ has size $a'$, this concludes the proof.
\end{proof}

\begin{lemma}
  \label{lemma0}
Let $m$ be a nonnegative integer, $t=2m+1$ and $a'\coloneqq\left\lfloor (|\mathcal{A}|-1)/(m+1)\right\rfloor$. If $C$ is unambiguous for $\Omega[\mathfrak{E}_t,\mA,\mF,S \to T,\mU_S,t]$, then $|C| \le a'$.
\end{lemma}
\begin{proof}
We follow the notation of the proof of Lemma \ref{lemma-1}. 
 Let $\mathcal{A}'\subseteq\mathcal{A}$ be an arbitrary subset of cardinality $a'$.
Let $a,b$ be arbitrary, distinct elements of $\mathcal{A}'$. 
Consider a vector $c_1$ of length~$2m+1$ such that $m$ of its coordinates are equal to $a$, and the remaining $m+1$ coordinates are equal to $b$. Similarly, consider a vector $c_2$ of length $2m+2$ such that $m+1$ of its coordinates are equal to $a$, and the remaining $m+1$ coordinates are equal to $b$. 

Note that the adversary can modify the codeword $(b,\dots,b)$ to the input $c_1$ for $V_1$ ($m$ modified edges) and the input $c_2$ for $V_2$ ($m+1$ modified edges). Then, the output state of $V_1$ is $(b,m+1)$ and the output state of $V_2$ is not characterized in the proof of Lemma~\ref{lemma-1}. To complete the proof, it suffices to show that $V_2$ cannot have an output state $(c,r)$ for any $c\in\mathcal{A'}$ and any $m+2\le r\le 2m+2$.

Assume that the output state of $V_2$ for input $c_2$ is $(a,i)$, where $i\ge m+2$. First we show that the output states $(b,m+1)$ for $V_1$ and $(a,i)$ for $V_2$ need to be decoded to the codeword~$(a,\dots,a)$. To this end, consider a vector $c_1'$ of length $2m+1$ such that $m$ of its coordinates are equal to~$a$, and the remaining $m+1$ coordinates are equal to $b$. Similarly, consider a vector $c_2'$ of length~$2m+2$ such that $i$ of its coordinates are equal to $a$, and the remaining $2m+2-i$ coordinates are equal to $b$. Now observe that the adversary can modify the codeword $(a,\dots,a)$ to the input $c_1'$ for $V_1$ ($m+1$ modified edges) and the input $c_2'$ for $V_2$ ($(2m+2-i)$ modified edges), where $m+1+(2m+2-i)\le 2m+1$ since $i\ge m+2$. Since the output state of $V_2$ is $(a,i)$ for $c_2'$ and the output state of $V_1$ is $(b,m+1)$, the decoder needs to map these pair of output states to $(a,\dots,a)$. However, the vectors $c_1$ for $V_1$ and $c_2$ for $V_2$ have the same outputs, by assumption, while the sent codeword is $(b,\dots,b)$. This is a contradiction as $a \neq b$. 

Finally, assume that the output state of $V_2$ for input $c_2$ is $(b,i)$, where $i\ge m+2$. Here we can interchange the roles of the symbols $a$ and $b$ in the previous argument, proving the result. 
\end{proof}

\begin{lemma}
\label{lemma1}
Let $m$ be a nonnegative integer and $t=2m+1$. We have $$\CC_1(\mathfrak{E}_t,\mA,\mU_S,t) \ge \log_{|\mA|}\left(\left\lfloor\tfrac{|\mathcal{A}|-1}{m+1}\right\rfloor\right).$$
\end{lemma}
\begin{proof}
Let $\mathcal{A}'\subseteq\mathcal{A}$ be an arbitrary subset of the alphabet of cardinality $a'\coloneqq\left\lfloor(|\mathcal{A}|-1)/(m+1)\right\rfloor$.
As outer code, we choose a $(2t+1)$-fold repetition code over the sub-alphabet $\mathcal{A}'$, i.e., $$\mC\coloneqq\left\{(c,\dots,c)\in\mathcal{A}^{2t+1}\,\mid\,c\in\mathcal{A}'\right\},$$ of cardinality $a'$.

For the function of $V_2$, we inject
the set of output states of $V_2$ into $\mathcal{A}$. There are $2m+2$ incoming edges of $V_2$. To have a unique majority winner that occurs more than half of those edges, call it $x$, the symbol $x$ must occur on at least $m+2$ of those edges. We have~$(2m+2)-(m+2)+1=m+1$. Thus, the injection we are looking for is of the form
$$\{!\} \cup (\mA'\times S) \to \mA,$$
where $|S|=m+1$ and ``!'' is one of the alphabet symbols. That implies that the size of the code is $$\left\lfloor\frac{\lvert \mathcal{A} \rvert-1}{\lvert S \rvert}\right\rfloor,$$ as desired. 

The decoding works as follows:
\begin{enumerate}
    \item The output state of $V_1$ is $!$ and the output state of $V_2$ is $(y,\beta)$. In this case, the answer is $(y,\ldots,y)$;
    \item The output state of $V_1$ is $(x,\alpha)$ and the output state of $V_2$ is $!$. In this case, the answer is $(x,\ldots,x)$;
    \item The output state of $V_1$ is $(x,\alpha)$ and the output state of $V_2$ is $(y,\beta)$. Here, we check whether $\alpha \ge \beta$. If so, we decode it to $(x,\ldots,x)$. Otherwise, we decode to $(y,\ldots,y)$.
\end{enumerate} 
For the analysis, assume that the $a$ incoming edges of $V_1$ and the $b$ incoming edges of $V_2$ are attacked, so that $a+b\le 2m+1$.
If $a\ge m+1$, then $V_2$ outputs $(y,\beta)$ for some $\beta > m+1$, and thus $y$ coincides with the coordinate of the originally sent codeword. So, if $V_1$ outputs $!$, then the decoding is correct. If $V_1$ outputs $(x,\alpha)$, then we can easily check $\alpha\le m+1 <\beta$ and the decoding is also correct.

Similarly, if $b\ge m+1$, then $V_1$ outputs $(x,\alpha)$ for some $\alpha \ge m+1$, and $x$ coincides with the coordinate of the originally sent codeword. So, if~$V_2$ outputs $!$, then the decoding is correct. If~$V_2$ outputs $(y,\beta)$, then one can easily check $\beta\le m+1 \le \alpha$ and the decoding is also correct. Note that we cannot have  
both $V_1$ and $V_2$ output $!$.

Finally, if $a \le m$ and $b<m+1$, then $V_1$ outputs $(x,\alpha)$ and $V_2$ outputs $(y,\beta)$, where we have~$\alpha,\beta>m+1$. In this case, we have $x=y$, which is the original sent codeword, so that the choice of the decoder does not matter.
\end{proof}

\begin{lemma}
\label{lemma2}
Let $m$ be a positive integer and $t=2m $. We have $$\CC_1(\mathfrak{E}_t,\mA,\mU_S,t) \ge \log_{|\mA|}\left(\left\lfloor\tfrac{|\mathcal{A}|}{m+1}\right\rfloor\right).$$
\end{lemma}
\begin{proof}
Let $a' = \left\lfloor |\mA|/ (m+1)\right\rfloor$ and label the codewords of $\mC$ as $c_1,\ldots,c_{a'}$, where $c_i = (i,\ldots,i)$ for all~$i \in \{1,\ldots,a'\}$. 
For $\beta \in \{m+1,\ldots,2m\}$ and $1 \le i \le a'$, the function of $V_2$ outputs~$(j,\beta)$, where the symbol $j$ occurs $\beta$ times on the incoming edges of $V_2$. If there is no symbol occurring more than $m+1$ times, the function of $V_2$ outputs $(j,m)$, where $j$ is any symbol of the alphabet. If there is a symbol, say $j$, occurring $2m+1$ times on the incoming edges of $V_2$, then~$V_2$ outputs~$(j,2m)$. Similarly $V_1$ sends $(i,\alpha)$, where $\alpha \in \{m+1,\ldots,2m\}$ and the symbol~$i$ occurs $\alpha$ times in the incoming edges of $V_1$. For the remaining possibilities, $V_1$ sends $(i,m)$, where $i$ is the most occurred symbol on the incoming edges of $V_1$. If it is not unique, $V_1$ picks the smallest one. 

We claim that the code $C$ with the described functions for the intermediate vertices $V_1$ and~$V_2$ creates an unambiguous pair and thus gives the desired lower bound.
Note that one needs different alphabet symbols for each element of the set $$\{(i,r): 1\le i \le a',\ m \le r \le 2m\}.$$ This is guaranteed by the assumption $a' = \left\lfloor \frac{|\mA|}{m+1}\right\rfloor$.

The decoding works as follows. Suppose that $T$ gets $(i,\alpha)$ from $V_1$ and $(j,\beta)$ from $V_2$. If~$\alpha \ge \beta$, the terminal decodes to $c_i$. Otherwise, $T$ decodes to $c_j$. For the rest of the proof, we explain why this decoding works.
\begin{enumerate}
    \item Assume that $\alpha=m$. That means there does not exist a symbol that occurs at least~$m+1$ times in the incoming edges of $V_1$. Thus, at least $m$ incoming edges of $V_1$ were attacked. Therefore at most $m$ of the incoming edges of $V_2$ were attacked, and thus the majority of the incoming edges of $V_2$ is maintained and $c_j$ should be decoded. Since $\beta \ge m+1 > m$, the decoding works correctly. 
    \item Let $\alpha \neq m$ and suppose that the terminal gets $(i,\alpha)$ and $(j,\beta)$. \begin{enumerate}
        \item One possibility is at least $\alpha$ of the incoming edges of $V_1$ were attacked and $i$ occurs~$\alpha$ times on the incoming edges of $V_1$. Thus, at most $2m-\alpha$ of the incoming edges of $V_2$ can be attacked. This number is always smaller than half of the number of incoming edges of $V_2$ since $\alpha \ge m+1$. Thus, $c_j$ should be decoded, and that is indeed the case because $\beta \ge (2m+1)-(2m-\alpha)=\alpha+1 > \alpha$.
        \item Assume that $x\le m$ edges among the incoming edges of $V_1$ were attacked. So,~$c_i$ should be decoded.  
        let~$y$ be the number of incoming edges of $V_2$ that were attacked and note that
        $y \le 2m-x$. If $\beta = y$ and~$y > 2m-x$, then $x+y > 2m=t$, a contradiction. In all the other cases where~$\beta =y$, the codeword $c_i$ is decoded by the decoding algorithm. The only case left to check 
        is $\beta = 2m+1-y$. We have~$\alpha = 2m-x < \beta$ if and only if $x \ge y$. However, in this case $j=i$, and thus we are done.
    \end{enumerate}
\end{enumerate}

Thus, we have shown that $\CC_1(\mathfrak{E}_{2m},\mA,\mU_S,2m) \ge \log_{|\mA|}\left(\left\lfloor\tfrac{|\mathcal{A}|}{m+1}\right\rfloor\right).$
\end{proof}

We can finally combine the four previous lemmas to compute the capacity of the networks of Family E.

\begin{theorem}
\label{thm:familyE}
Let $m$ be a nonnegative integer and $t=2m + \alpha$, where $\alpha \in \{0,1\}$. We have
$$\CC_1(\mathfrak{E}_t,\mA,\mU_S,t) =  \log_{|\mA|}\left(\left\lfloor\tfrac{|\mA|-\alpha}{m+1}\right\rfloor\right).$$
\end{theorem}
\begin{proof}
The case where $t$ is even follows from Lemma \ref{lemma-1} and Lemma \ref{lemma2}, and the case where~$t$ is odd follows from Lemma \ref{lemma0} and \ref{lemma1}.   
\end{proof}
For example, when $t=1$, Theorem \ref{thm:familyE} gives the capacity of the Diamond Network; see Theorem \ref{thm:diamond}. We conclude the section with the next remark.

\begin{table}[h!]
\centering
\begin{tabular}{ |p{4cm}|p{2cm}|  }
\hline
\multicolumn{2}{|c|}{Largest size of an unambiguous outer code for Family E} \\
\hline \hline
$t=1$ & $\lvert \mA \rvert$-1 \\
\hline
$t=2$ & $\lvert \mA \rvert/2$ \\
\hline
$t=3$ & $(\lvert \mA \rvert-1)/2$ \\
\hline
$t=4$ & $\lvert \mA \rvert/3$ \\
\hline
$t=5$ & $(\lvert \mA \rvert-1)/3$ \\
\hline
\end{tabular}
\caption{Table for Remark \ref{rem:familE}.}
\label{tab:diam}
\end{table}

\begin{remark}
\label{rem:familE}
Theorem \ref{thm:familyE} computes the capacity of any member of Family E. Naturally, as~$t$ increases, the capacity approaches 0. 
Roughly speaking, the largest size of unambiguous outer codes for $t=2m$ and $t=2m+1$ is the same, and this number becomes $\frac{m}{m+1}$ when~$t=2m-1$ for~$m \ge 1$. Now, observe that $\frac{m}{m+1}$ attains its smallest value when $m=1$, showing the  discrepancy between the Diamond Network and the remaining members of Family E.
We refer to Table \ref{tab:diam} for the first few instances to highlight this phenomenon. 
\end{remark}

\section{A Generalized Family of Simple 2-Level Networks}
\label{sec:newfamily}

In this section, we define a new family of networks that contains various interesting families, including Family B. This family also showcases some interesting phenomena of network coding for specific parameters. Let $a,s \ge 1$ and $b \ge 0$ be integers. Define $$\mathfrak{S}_{a,b,s} = ([a,b+s],[a,s]);$$ see Figure \ref{fig:newfamily} for an illustration. 

\begin{figure}[htbp]
\centering

\begin{tikzpicture}
\tikzset{vertex/.style = {shape=circle,draw,inner sep=0pt,minimum size=1.9em}}
\tikzset{nnode/.style = {shape=circle,fill=myg,draw,inner sep=0pt,minimum
size=1.9em}}
\tikzset{edge/.style = {->,> = stealth}}
\tikzset{dedge/.style = {densely dotted,->,> = stealth}}
\tikzset{ddedge/.style = {dashed,->,> = stealth}}

\node[vertex] (S1) {$S$};

\node[shape=coordinate,right=\mynodespace of S1] (K) {};

\node[nnode,above=0.8\mynodespace of K] (V1) {$V_1$};
\node[nnode,below=0.8\mynodespace of K] (V2) {$V_2$};

\node[vertex,right=2\mynodespace of S1] (T) {$T$};

\draw[draw= none,bend left=5] (S1)  to node[sloped,fill=white, inner sep=3pt]{\small $\vdots$} (V1);
\draw[ddedge,bend left=25] (S1) to  node[fill=white, inner sep=5pt]{\small $e_1$} (V1);
\draw[ddedge,bend right=25] (S1) to  node[fill=white, inner sep=5pt]{\small $e_a$} (V1);

\draw[draw= none,bend left=5] (S1)  to node[sloped,fill=white, inner sep=5pt]{$\vdots$} (V2);
\draw[ddedge,bend left=25] (S1) to  node[fill=white, inner sep=5pt]{\small $e_{a+1}$} (V2);
\draw[ddedge,bend right=25] (S1) to  node[fill=white, inner sep=5pt]{\small $e_{a+s+b}$} (V2);

\draw[draw= none,bend left=5] (V1)  to node[sloped,fill=white, inner sep=5pt]{$\vdots$} (T);
\draw[edge,bend left=25] (V1) to  node[fill=white, inner sep=5pt]{\small $e_{a+s+b+1}$} (T);
\draw[edge,bend right=25] (V1) to  node[fill=white, inner sep=5pt]{\small $e_{2a+s+b}$} (T);

\draw[draw= none,bend left=5] (V2)  to node[sloped,fill=white, inner sep=5pt]{$\vdots$} (T);
\draw[edge,bend left=25] (V2) to  node[fill=white, inner sep=5pt]{\small $e_{2a+s+b+1}$} (T);
\draw[edge,bend right=25] (V2) to  node[fill=white, inner sep=5pt]{\small $e_{2a+2s+b}$} (T);

\end{tikzpicture} 

\caption{A network of the new family $\mathfrak{S}_{a,b,s}$. }\label{fig:newfamily}
\end{figure}
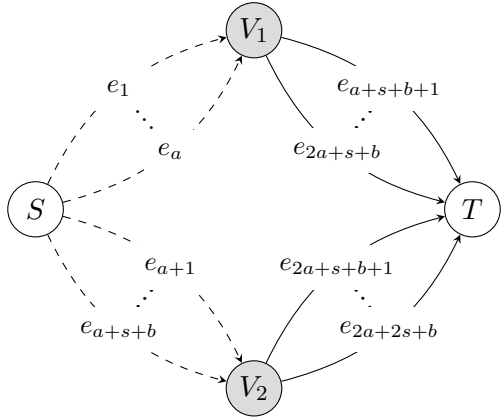

Theorem \ref{thm:sbound} reads 
\begin{equation}
\label{eq:thm19_reads}
\CC_1(\mathfrak{S}_{a,b,s},\mA,\mU_S,1) \le
\begin{cases}      
      a+s-2  & \text{if } b=0, \\
      s + \max\{0,a-2\}  & \text{otherwise}. \\
\end{cases}
\end{equation}

The network family specializes to 
other network families by choosing some specific parameters, as the next result illustrates.

\begin{proposition}
\label{prop:newfamily}
We have 
\begin{enumerate}
    \item \label{item1} $\mathfrak{S}_{a,b,s} = \mathfrak{B}_s$ when $a =b=1$.
    \item \label{item2} $\CC_1(\mathfrak{S}_{a,0,s},\mA,\mU_S,1) = A_{\mA}(a+s,3)$.
    \item \label{item3} $\mathfrak{S}_{a,b,s}$ coincides with the Diamond Network of Theorem \ref{thm:diamond} when $a=b=s=1$, and thus $\CC_1(\mathfrak{S}_{1,1,1},\mA,\mU_S,1) = \log_{|\mA|}(\lvert \mA \rvert-1).$ Note that for any triple $(1,b,1)$ with $b>1$, we have $\CC_1(\mathfrak{S}_{1,b,1},\mA,\mU_S,1)=1$.
    \item \label{item4} $\CC_1(\mathfrak{S}_{2,1,1},\mA,\mU_S,1) = 1$.
    \item \label{item5} $\CC_1(\mathfrak{S}_{3,b,1},\mA,\mU_S,1) = 
    1$ if $(b,\lvert \mA \rvert)=(1,2)$, and 
    $\CC_1(\mathfrak{S}_{3,b,1},\mA,\mU_S,1) =2$ otherwise.
\end{enumerate}
\end{proposition}
\begin{proof}
We provide the proofs of the above statements in the same order. For the corresponding upper bounds, we always refer to the bound \eqref{eq:thm19_reads}. 
\begin{enumerate}
    \item The statement follows from the definition of Family B.
    \item When $b=0$, both vertices have the same number of incoming and outgoing edges. As explained in the beginning of Section \ref{sec:threefam}, both functions can be assumed to be the identity. This then transforms the question into the classical coding theory problem of determining the maximum size of a code of length $a+s$ and minimum distance at least $3$ for the given alphabet size.  
    \item The first part follows from the definition of the Diamond Network. If $(a,b,s) = (1,b,1)$ with $b >1$, then one could send a repetition code along the incoming edges of $V_2$ and put a majority decoder at $V_2$. Since $b+1\ge 3$, we get the desired result.
    \item Let $\mC = \{(a,a,a,a): a \in \mA\}$, and let $\mF_{V_1}$ be the identity function and let $\mF_{V_2}$ be the function that forwards the incoming packet on $e_3$. For any action of the adversary, majority decoding at $T$ achieves the corresponding upper bound for any alphabet size.
    \item Assume $b \neq 1$. Then the corresponding upper bound $2$ is met: The source sends a $[3,1,3]$ code to~$V_1$, and $V_1$ uses the identity function. The terminal $T$ decodes and gets one of the symbols. Similarly, the source sends a $[3,1,3]$ code along the edges $e_4$, $e_5$,  $e_6$ and some fixed symbols along the other~$b+1-3$ (recall $b\ge 2$) incoming edges of $V_2$. The function of $V_2$ decodes the code that is sent along $e_4$, $e_5$, $e_6$ and sends it to the terminal so that the terminal gets the second symbol as desired. 
    
    Assume that $\lvert \mA \rvert \ge 3$. Then, there exists a $[4,2,3]$ MDS code over $\mA$. The source sends it along $e_1,e_2,e_3$ and $e_4$ and repeats fixed alphabet symbols along the other $b+1-3$. The node $V_1$ uses the identity function and the function at $V_2$ only forwards what is on the edge $e_4$. The terminal then simply uses the decoder of that MDS code. 
    
    Finally, assume $(b,\lvert \mA \rvert)=(1,2)$, the only remaining case. The function at $V_1$ can again be assumed to be the identity function. For the function at $V_2$, there are~$2^4=16$ options. There are $120$ different error correcting codes of length $5$ and minimum distance $3$. Simply, going through all the $16 \cdot 120 = 1920 $ pairs, one concludes that  $\CC_1(\mathfrak{S}_{3,1,1},\mA,\mU_S,1) = 1$ when $\lvert \mA \rvert=2$. \qedhere
\end{enumerate}
\end{proof}

\begin{remark}
When $b=0$, Part~\ref{item2} in Proposition \ref{prop:newfamily} shows that the two intermediate nodes can be combined into one. Even further, that node can be disregarded as forwarding simply achieves the upper bound of \eqref{eq:thm19_reads}. This can also be seen with the upper bound on~$\CC_1(\mathfrak{S}_{a,b,s},\mA,\mU_S,1)$ in the case $b=0$. Although the upper bound on~$\CC_1(\mathfrak{S}_{a,b,s},\mA,\mU_S,1)$ only depends on whether $b$ is zero or not, the value of $b$ does affects the capacity; see, for example, Part~\ref{item3} of Proposition \ref{prop:newfamily}. 

Part~\ref{item1} of Proposition~\ref{prop:newfamily} already proves that the upper bound of \eqref{eq:thm19_reads} is not always met for any choice of parameters. However, there are various cases, given in Proposition \ref{prop:newfamily}, where the bound of \eqref{eq:thm19_reads} is met with equality; see, for example, Part~\ref{item4}. 

Lastly, there are cases where the bound of \eqref{eq:thm19_reads} is always met except for a single choice of parameters; see, for instance, Part~\ref{item5} of Proposition \ref{prop:newfamily}.
\end{remark}

In the next result, we rewrite \eqref{eq:thm19_reads} and remove the ``maximum" in the upper bound.
\begin{proposition}
\label{prop:updateupper}
We have $$\CC_1(\mathfrak{S}_{a,b,s},\mA,\mU_S,1) \le
\begin{cases}      
      s,  & \text{if } b\neq 0,\  a\le 2, \mbox{ and }\lvert \mA \rvert < s+1, \\
      a+s-2,  & \text{if } b\neq 0,\  a > 2, \mbox{ and }\lvert \mA \rvert < a+s-1,
\end{cases}$$ and
$$\CC_1(\mathfrak{S}_{a,b,s},\mA,\mU_S,1) =
\begin{cases}
      \log_{|\mA|}A_{\mA}(a+s,3),  & \text{if } b= 0, \\
      s,  & \text{if } b\neq 0,\  a\le 2 \mbox{ and }\lvert \mA \rvert \ge s+1, \\
      a+s-2,  & \text{if } b\neq 0,\  a > 2 \mbox{ and }\lvert \mA \rvert \ge a+s-1.
\end{cases}$$
\end{proposition}
\begin{proof}
The only cases that were not already solved in this section so far are:
\begin{itemize}
    \item $b\neq 0,\  a \le 2,\ |\mA| \ge s+1$,
    \item $b\neq 0,\  a > 2 \mbox{ and }|\mA| \ge a+s-1$.
\end{itemize}
In those cases, there exists an $[s+2,s,3]$ MDS code over $\mA$ and an $[a+s,a+s-2,3]$ MDS code over $\mA$, respectively. Then, one simply forwards at the intermediate vertices ignoring the extra edges and decodes at the terminal $T$. This works as $a+b+s \ge s+2$, since $b \neq 0$ and~$1 \le a \le 2$ in the first case, and $a+b+s \ge a+s$ (which is trivially correct) in the second case.
\end{proof}

We conclude this section by solving one of the cases where the network capacity is not known. As written in Proposition~\ref{prop:updateupper}, the value $\CC_1(\mathfrak{S}_{a,b,s},\mA,\mU_S,1)$ is unknown, in general, if~$b \neq 0$, $a>2$ and $|\mA| < a+s-1.$

\begin{theorem}
\label{thm:another_interesting_ex}
Let $(a,b,s)=(3,1,2)$. We have 
\begin{enumerate}
    \item $\CC_1(\mathfrak{S}_{3,1,2},\F_2,\mU_S,1) = \log_2 6$.
    \item $3 \ge \CC_1(\mathfrak{S}_{3,1,2},\F_3,\mU_S,1) \ge \log_3 15$.
    \item $\CC_1(\mathfrak{S}_{3,1,2},\mA,\mU_S,1) = 3$ if $\lvert \mA \rvert \ge 4$.
\end{enumerate}
\end{theorem}
\begin{proof}
For the first two items, we refer the reader to Appendix~\ref{sec:concretecodes}.
There, we provide for both items an unambiguous code of size~6 and~15, respectively.
Moreover, to show that the exact value of the capacity in Part~1 is~$\log_2 6$, we used a computer-based proof that no code of size~7 exists; see Appendix~\ref{sec:concretecodes} for more details.\footnote{The code used in our experiments is available at~\cite{zenodo} and \href{https://github.com/christopherhojny/supplement_network_codes}{GitHub}.} Note that we were also able to prove Part~1 ``by hand''; however, the proof is not instructive and very cumbersome, thus we choose not to inflict it upon the reader. In words, we enumerated all binary codes of minimum distance 3 of size at least 6 up to symmetry. That gives one possible code of size 8 and three codes of size 7. The proof is then concluded by showing impossibility results for all of these four codes separately. 

For the last part, Proposition \ref{prop:updateupper} gives $\CC_1(\mathfrak{S}_{3,1,2},\mA,\mU_S,1) \le 3$. This upper bound can be achieved simply by forwarding at the intermediate nodes and sending an $[5,3,3]$ MDS code, whose existence is guaranteed by the assumption on the size of the alphabet.
\end{proof}

Theorem \ref{thm:another_interesting_ex} yet provides another example where the upper bound of Theorem \ref{thm:sbound} is met with equality over large alphabet sizes. However, the bound is not attained and requires more computationally involved techniques to fully compute over small alphabet sizes.

\section{Separability}
\label{sec:separable}
The main challenge in establishing the capacity of networks with restricted adversaries appears to be the necessity of designing inner and outer codes \textit{jointly}. This is not the case when the adversary is not restricted to operate on a proper subset of the edges. In such cases, a linear network (inner) code in combination with any rank metric (outer) code with the appropriate parameters can achieve capacity; see~\cite{silva2008rank} and the references therein. 

The goal of this section is to introduce a formal notion of  ``separability'' of a network with respect to a metric for the outer code. We will then be able to formally prove that networks with unrestricted adversaries are separable with respect to the rank metric ($\drk$), whereas networks with restricted adversaries are not. We will also prove that networks with unrestricted adversaries are in general not separable with respect to the Hamming metric ($\dH$).

\begin{definition}
    Let $\mA$ be an alphabet and let $n \ge 1$ be an integer. A \textbf{metric} on $\mA^n$ is a function $d\colon\mA^n \times \mA^n \to \R$ with the following properties:
    \begin{enumerate}
        \item $d(x,y) \ge 0$, with equality if and only if $x=y$ for all $x,y \in \mA^n$;
        \item $d(x,y)=d(y,x)$ for all $x,y \in \mA^n$;
        \item $d(x,z) \le d(x,y)+d(y,z)$ for all $x,y,z \in \mA^n$.
\end{enumerate}
A \textbf{code} is a nonempty subset $\mC \subseteq \mA^n$. We say that $\mC$ \textbf{corrects} $t$ \textbf{errors} with respect to the metric $d$ if the following holds: 
$$\textnormal{For every } x \in \mC \mbox{ and } y \in \mA^n \mbox{ with } d(x,y) \le t,\ \textnormal{we have } \{z \in \mC \mid d(z,y) \le t\}=\{x\}.$$
\end{definition}

\begin{definition}
   Let $\mN=(\mV,\mE, S, \bfT)$ be a network, $\mA$ be an alphabet, $\mU \subseteq \mE$ an edge set with~$\mU \supseteq \out(S)$, and $t \ge 0$ an integer. Let $n=\degout(S)$ and let $d\colon \mA^n \times \mA^n \to \R$ be a metric. We say that $(\mN,\mA,\mU,t)$ is \textbf{separable} with respect to $d$ if there exists a network code $\mF$ for~$(\mN,\mA)$ with the following property: If 
   $\mC \subseteq \mA^n$ corrects $t$ errors and $\smash{\log_{|\mA|}(\mC)=\lvert \mA \rvert^{\CC_1(\mN,\mA,\mU,t)}}$, then $\mC$ is unambiguous for the channel $\Omega[\mN,\mA,\mF,S\to T,\mU,t]$ for every $T \in {\bf T}$.   
\end{definition}

In words, the previous definition declares a network separable when it admits a ``universal'' network code $\mF$. When such a network code is used, the source can design the outer code $\mC$ locally, and any local outer code will extend to a global outer code. In particular, the source may not be unaware of the network code $\mF$.

It is well known that if
$\mN=(\mV,\mE, S, \bfT)$ is a network with alphabet $\mA=\F_{q^m}$, then~$(\mN,\mA,\mE,t)$ is separable with respect to the rank metric (note that the adversary is unrestricted).
A proof of this fact can be found in~\cite{dikaliotis2011multiple} following the main ideas of \cite{silva2008rank}. 

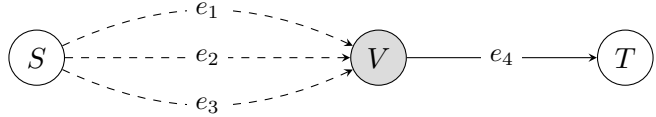
\begin{figure}[h!]
\centering

\begin{tikzpicture}
\tikzset{vertex/.style = {shape=circle,draw,inner sep=0pt,minimum size=1.9em}}
\tikzset{nnode/.style = {shape=circle,fill=myg,draw,inner sep=0pt,minimum
size=1.9em}}
\tikzset{edge/.style = {->,> = stealth}}
\tikzset{dedge/.style = {densely dotted,->,> = stealth}}
\tikzset{ddedge/.style = {dashed,->,> = stealth}}

\node[vertex] (S1) {$S$};

\node[shape=coordinate,right=\mynodespace*0.5 of S1] (K) {};

\node[nnode,right=\mynodespace of K] (V) {$V$};

\node[vertex,right=\mynodespace of V] (T) {$T$};

\draw[ddedge,bend left=25] (S1)  to node[fill=white, inner sep=3pt]{\small $e_1$} (V);

\draw[ddedge,bend right=25] (S1)  to node[fill=white, inner sep=3pt]{\small $e_3$} (V);

\draw[ddedge,bend right=0] (S1)  to node[fill=white, inner sep=3pt]{\small $e_2$} (V);

\draw[edge,bend left=0] (V)  to node[fill=white, inner sep=3pt]{\small $e_{4}$} (T);
\end{tikzpicture}
\caption{Network for Proposition \ref{prop:restr}.\label{easyex}}
\end{figure}

The next result shows that networks with 
restricted adversaries
are not separable, in general, with respect to the Hamming and the rank metric. 

\begin{proposition} \label{prop:restr}
    Let $\mN=(\mV,\mE, S, \bfT)$ be the network in Figure~\ref{easyex}, and let $\mU=\{e_1,e_2,e_3\}$. Let the alphabet $\mA$ be arbitrary. Then $(\mN,\mA,\mU,1)$ is not separable with respect to $\dH$ and~$\drk$, where $\dH$ is the Hamming metric and $\drk$ is the rank metric.
\end{proposition}

\begin{proof}
We have $\CC_1(\mN,\mA,\mU,1)=1$. Let $\mC$ be a code that corrects 1 error with $\lvert \mC \rvert=\lvert \mA \rvert$. Given~$c \in \mC$, we denote by $B_1(c)$
the ball of radius 1 with center $c$ with respect to the respective metric. The function of $V$ must be a ``decoder" of $\mC$ in the sense that the following must hold:
For any $x\in B_1(c)$ and $y\in B_1(c')$ with $c \neq c'$, we have $\mF_V(x) \neq \mF_V(y)$. If not, this would create a collision as $d(x,c) \le 1$ and~$d(y,c') \le 1$. Now the result follows from the fact that there does not exist a universal decoder in the above sense that works for all codes in the Hamming or rank metric. 
\end{proof}

Note that $(\mN,\mA,\mE,1)$ is separable with respect to the rank metric, where $\mN$ and $\mA$ are as in Proposition \ref{prop:restr}. 

It is natural to ask if networks with unrestricted adversaries are separable with respect to metrics that are different from the rank metric. An obvious candidate is the Hamming metric. However, the next result shows that, in general,
networks are not separable with respect to the Hamming metric, not even when the adversary is unrestricted.

\begin{figure}[h!]
\centering

\begin{tikzpicture}

\tikzset{vertex/.style = {shape=circle,draw,inner sep=0pt,minimum size=1.9em}}
\tikzset{nnode/.style = {shape=circle,fill=myg,draw,inner sep=0pt,minimum
size=1.9em}}
\tikzset{edge/.style = {->,> = stealth}}
\tikzset{dedge/.style = {densely dotted,->,> = stealth}}
\tikzset{ddedge/.style = {dashed,->,> = stealth}}

\node[vertex] (S1) {$S$};

\node[shape=coordinate,right=\mynodespace of S1] (K) {};

\node[nnode,above=0.8\mynodespace of K] (V1) {$V_1$};

\node[nnode,below=0.8\mynodespace of K] (V3) {$V_3$};

\node[nnode,right=0.1\mynodespace of K] (V2) {$V_2$};

\node[vertex,right=2.2\mynodespace of V1 ] (T1) {$T_1$};

\node[vertex,right=2.2\mynodespace of V3] (T2) {$T_2$};

\draw[ddedge,bend left=0] (S1)  to node[fill=white, inner sep=3pt]{\small $e_1$} (V1);

\draw[ddedge,bend left=0] (S1) to  node[fill=white, inner sep=3pt]{\small $e_2$} (V2);

\draw[ddedge,bend right=0] (S1) to  node[fill=white, inner sep=3pt]{\small $e_3$} (V3);

\draw[ddedge,bend left=0] (V1) to  node[fill=white, inner sep=3pt]{\small $e_4$} (V2);

\draw[ddedge,bend left=0] (V3) to  node[fill=white, inner sep=3pt]{\small $e_5$} (V2);

\draw[ddedge,bend left=0] (V1) to  node[fill=white, inner sep=3pt]{\small $e_6$} (T1);

\draw[ddedge,bend left=0] (V3) to  node[fill=white, inner sep=3pt]{\small $e_9$} (T2);

\draw[ddedge,bend left=10] (V2) to  node[fill=white, inner sep=3pt]{\small $e_{7}$} (T1);

\draw[ddedge,bend right=10] (V2) to  node[fill=white, inner sep=3pt]{\small $e_{8}$} (T1);

\draw[ddedge,bend left=10] (V2) to  node[fill=white, inner sep=3pt]{\small $e_{9}$} (T2);

\draw[ddedge,bend right=10] (V2) to  node[fill=white, inner sep=3pt]{\small $e_{10}$} (T2);

\end{tikzpicture} 

\caption{Network for Proposition~\ref{prop:notH}.}
\label{figex}
\end{figure}
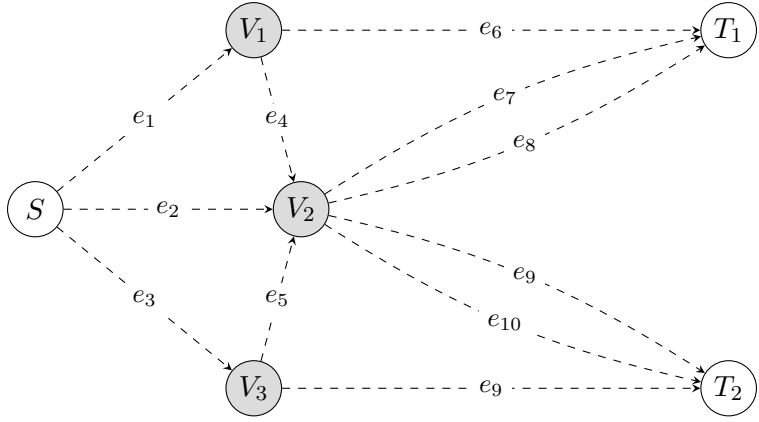

\begin{proposition} \label{prop:notH}
    Let $\mN=(\mV,\mE, S, \bfT)$ be the network in Figure~\ref{figex}. Then $(\mN,\F_2,\mE,1)$ is not separable with respect to $\dH$.
\end{proposition}

\begin{proof}
One can check that $\CC_1(\mN,\F_2,\mE,1) = 1$. There are 4 possible outer codes, namely:
\begin{itemize}
    \item $\{(0,0,0),(1,1,1)\}$,
    \item $\{(1,0,0),(0,1,1)\}$,
    \item $\{(0,1,0),(1,0,1)\}$,
    \item $\{(0,0,1),(1,1,0)\}$.
\end{itemize}
Without loss of generality, we have $\mF_{V_1}(x)=\mF_{V_3}(x)=x$ for $x \in \F_2$ as $\degin(V_1)=\degin(V_3)=1$, see \cite[Proposition 2]{hojny2023role}. Now we study the bottleneck at $V_2$. Let 
$$\mF_{V_2}(x,y,z) = (f(x,y,z),g(x,y,z))$$ with $f,g\colon\mA^3 \rightarrow \mA^2$ such that $f$ outputs along $e_7,e_8$ and $g$ outputs along $e_9,e_{10}$. 
The value~$f(x,y,z)$ cannot be equal to $f(1-x,1-y,1-z)$ for $(x,y,z) \in \F_2^3$ due to the structure of all possible outer codes (case-by-case analysis). The same holds for $g$. This shows that the functions $f$ and $g$ have the property that their outputs evaluated at 
$$(0,0,0), (0,1,1),(1,0,1) \mbox{ and }(1,1,0)$$ are pairwise different.

No matter which inner code is selected, the outer code $\{(0,0,0),(1,1,1)\}$ can be used to create a collision as follows: Assume $f(1,1,1)=(a,b)$. We then have three cases:
\begin{enumerate}[label=(\roman*)]
    \item $f(0,0,0)=(a,c)$ with $b\neq c$. The adversary acts on $e_{8}$, changes $b$ to $c$ when $(1,1,1)$ is sent, so that $T_1$ receives $(1,a,c).$ The adversary acts on $e_{6}$, changes $0$ to $1$ when $(0,0,0)$ is sent, so that $T_1$ receives $(1,a,c).$
    \item $f(0,0,0)=(c,b)$ with $a\neq c$. The adversary acts on $e_{7}$, changes $a$ to $c$ when $(1,1,1)$ is sent so that $T_1$ receives $(1,c,b).$ The adversary acts on $e_{6}$, changes $0$ to $1$ when $(0,0,0)$ is sent, so that $T_1$ receives $(1,c,b).$
    \item $f(0,0,0)=(b,a)$ with $a\neq b$. We know that $f(0,1,1), f(1,1,0)$ and $f(0,0,0)$ are pairwise distinct. Thus, either $$\dH(f(0,0,0),f(0,1,1)) \le 1 \mbox{ or } \dH(f(0,0,0),f(1,1,0)) \le 1.$$ Let $\dH(f(0,0,0),f(0,1,1)) \le 1$. (The other case would use the function $g$ and leads to a collision at $T_2$). The adversary acts on $e_{1}$, changes $1$ to $0$ when $(1,1,1)$ is sent so that~$T_1$ receives $(0,f(0,1,1)).$ The adversary acts on $e_{7}$ or $e_{8}$ when $(0,0,0)$ is sent, so that $T_1$ receives~$(0,f(0,1,1))$ which follows from the fact that~$\dH(f(0,0,0),f(0,1,1)) \le 1$. \qedhere    
\end{enumerate} 
\end{proof}

\section{Conclusions and Future Research Questions}
\label{sec:conclusion}

We computed the capacity of one of the families introduced in \cite{beemer2023network} and improved some cases of \cite{kurz2022capacity} on Family B. Furthermore, we introduced a general family of networks in Section \ref{sec:newfamily} that contains Family B and computed its capacity for certain parameters and explained the interesting phenomena that they are examples of in network coding. The following research question naturally originates from this: 

\begin{enumerate}
    \item Can computing the capacity of networks in Section \ref{sec:newfamily} for more parameters and the techniques used therein lead to greater understanding on the capacity of networks of Family B?  
\end{enumerate}

In addition to the fact that classical approaches to compute the capacity of networks with unrestricted adversaries are not even close to being optimal in the case of restricted adversaries (see \cite{beemer2023network}), we note in this paper that one of the main difficulties in estimating or achieving the capacity of networks with restricted adversaries is the inevitability of designing outer and inner codes together. To this end, Section \ref{sec:separable} studies the separability of networks. We only focused on rank and Hamming metric, and gave examples of networks where the network is not separable with respect to both of them, or only separable with respect to the rank metric. The latter is already a well-known fact for networks with unrestricted adversaries. We conclude the paper with a future research question arising from that.

\begin{enumerate}\setcounter{enumi}{1}
    \item Provide a criterion for a network to be separable with respect to the Hamming metric. 
\end{enumerate}

\bibliographystyle{abbrv}
\bibliography{main}

\appendix

\section{A MINLP Model for Finding Network Codes}
\label{sec:MINLPmodel}

In this appendix, we explain how to find codes for the newly
introduced family of networks of Section~\ref{sec:newfamily} by means of
mixed-integer programming. To keep the notation simple, we focus in our presentation on simple 2-level networks (see Definition \ref{def:2-level}) with two intermediate nodes, where the adversarial power is equal to 1. However, we remark that all ideas naturally generalize to other networks. Moreover, at the end of this section, we sketch how to adapt our mixed-integer programming techniques to the SAT framework, which was used to derive some of the results presented in Proposition~\ref{prop:familyB}.

A mixed-integer nonlinear program (MINLP) is a mathematical optimization
problem of the form
\[
  \min\{ f(x) : g(x) \leq 0,\; x \in \Z^p \times \R^{n-p}\},
\]
where~$p$ and~$n$ are non-negative integers satisfying~$p \leq n$ and~$f\colon \R^n \to \R$,
$g\colon \R^n \to \R^m$.
The inequality~$g(x) \leq 0$ is interpreted componentwise.
If both~$f$ and~$g$ are affinely linear functions, the problem is called a
mixed-integer linear program (MILP).
In the past decennia, powerful algorithms have been developed to solve
MINLP and MILP problems, and the general nature of MINLPs and MILPs allows
to model and solve problems of different kinds such as graph coloring~\cite{CollEtAl2002}
or network design problems~\cite{GounarisEtAl2016}.
In a previous publication~\cite{hojny2023role}, we also discussed how to find robust network
codes by means of MINLP. We use the following notation.
\begin{notation}
For~$v \in \{V_1,V_2\}$, let~$D^{\text{in}}_v$ and~$D^{\text{out}}_v$ be
the edges entering and leaving~$v$, respectively.
\end{notation}

Suppose that the goal is to decide whether there exists an unambiguous outer
code of size~$M$ for an alphabet of size~$\ell$.
Let~$\mathcal{C} = \{c_1,\dots,c_M\}$ and~$\mathcal{A} = \{1,\dots,\ell\}$.
Both the outer code and the inner code will be modeled as variables in our MINLP model.
Let~$v \in \{V_1,V_2\}$.
To model the function at~$v$, we define for all tuples~$\sigma
\in \mathcal{A}^{D^{\text{in}}_v}$ and~$\tau \in
\mathcal{A}^{D^{\text{out}}_v}$ a binary variable~$f^v_{\sigma,\tau}$ which
shall have the following meaning:
\[
  f^v_{\sigma,\tau}
  =
  \begin{cases}
    1, & \text{if the function at~$v$ transforms input~$\sigma$ into output~$\tau$},\\
    0, & \text{otherwise}.    
  \end{cases}
\]
Moreover, for every edge~$a \in D^{\text{in}}_{V_1} \cup
D^{\text{out}}_{V_2}$, symbol~$s \in \mathcal{A}$, and codeword~$c \in
\mathcal{C}$, we introduce a binary variable~$x^c_{a,s}$ which shall have
the following meaning:
\[
  x^c_{a,s}
  =
  \begin{cases}
    1, & \text{if codeword~$c$ assigns edge~$a$ the symbol~$s$ in the absence of
         attacks},\\
    0, & \text{otherwise}.
  \end{cases}
\]

Next to these ``natural'' decision variables, we also introduce some
auxiliary variables. Let~$\Delta$ denote the set of admissible attacks where each element of $\Delta$ 
is the set of attacked edges.
For~$c \in \mathcal{C}$, $v \in \{V_1,V_2\}$, $\delta \in \Delta$, and~$\sigma \in
\mathcal{A}^{D^{\text{in}}_v}$, we introduce a binary
variable~$y^{c,v}_{\sigma,\delta}$ which shall have the following meaning:

\[
  y^{c,v}_{\sigma,\delta}
  =
  \begin{cases}
    1, & \begin{aligned}
         & \text{if the input to node~$v$ from codeword~$c$ on the edges not in $\delta$} \\
         & \text{agrees with~$\sigma$ when restricted to the unattacked edges.}
         \end{aligned} \\
    0, & \text{otherwise}.
  \end{cases}
\]

Finally, for~$c \in \mathcal{C}$, $v \in \{V_1,V_2\}$, $\delta \in \Delta$,
and~$\tau \in \mathcal{A}^{D^{\text{out}}_v}$, we introduce a binary variable~$z^{c,v}_{\tau,\delta}$ which shall have the following meaning: 

\[
  z^{c,v}_{\tau,\delta}
  =
  \begin{cases}
    1, & \text{if the
(manipulated) codeword~$c$ sends input~$\tau$ to the terminal~$T$ under
attack~$\delta$}, \\
    0, & \text{otherwise}.
  \end{cases}
\]

To provide all variables with the aforementioned meaning, we need to introduce
some constraints that link them.
The first family of constraints makes sure that every codeword~$c \in
\mathcal{C}$ assigns every edge~$a \in D^{\text{in}}_{V_1} \cup
D^{\text{in}}_{V_2}$ exactly one symbol.
Since all variables are binary, this can be achieved by the condition
\[
  \sum_{s \in \mathcal{A}} x^c_{a,s} = 1.
\]

The linking of~$x$-variables and~$y$-variables is achieved by the following
constraints.
For all codewords $c \in \mathcal{C}$,~$v \in \{V_1,V_2\}$, $\sigma \in
\mathcal{A}^{D^{\text{in}}_v}$, and~$\delta \in \Delta$, let~$M^v_\delta$
be the edges entering~$v$ that are \emph{not} attacked by~$\delta$
and define the constraints
\begin{align*}
  \sum_{a \in M^v_\delta} x^c_{a,\sigma_a}
  &\leq
    y^{c,v}_{\sigma,\delta} + \lvert M^v_\delta \rvert - 1 &&\\
  y^{c,v}_{\sigma,\delta}
  &\leq
    x^c_{a,\sigma_a}, && a \in M^v_\delta.
\end{align*}
In fact, the first inequality forces~$y^{c,v}_{\sigma,\delta}$ to take
value~1 if all unattacked edges~$a$ are assigned symbol~$\sigma_a$, whereas
the second family of constraints forces~$y^{c,v}_{\sigma,\delta}$ to~0 when
codeword~$c$ does not assign an unattacked edge~$a$ symbol~$\sigma_a$, i.e.,
$x^c_{a,\sigma_a} = 0$.

To make sure that the~$f$-variables indeed model functions, we enforce for
every~$v \in \{V_1,V_2\}$ that every input~$\sigma \in
\mathcal{A}^{D^{\text{in}}_v}$ has exactly one image~$\tau \in
\mathcal{A}^{D^{\text{out}}_v}$ via
\[
  \sum_{\tau \in \mathcal{A}^{D^{\text{out}}_v}} f^v_{\sigma,\tau} = 1.
\]

The values of the~$z^{c,v}$-variables, i.e., possible information received
at the terminal~$T$ from node~$v$, is fully determined by the possible inputs
to node~$v$ (the~$y^v$-variables) and the function modeled via~$f^v$.
This relation can be modeled as follows.

For~$c \in \mathcal{C}$, $v \in \{V_1,V_2\}$, $\sigma \in
\mathcal{A}^{D^\text{in}_v}$, $\tau \in \mathcal{A}^{D^{\text{out}}_v}$,
and~$\delta \in \Delta$, we enforce
\begin{align*}
  z^{c,v}_{\tau,\delta}
  &\leq
    \sum_{\sigma \in \mathcal{A}^{D^{\text{in}}_v}} f^v_{\sigma,\tau} \cdot y^{c,v}_{\sigma,\delta}, &&\\
  z^{c,v}_{\tau,\delta}
  &\geq
    f^v_{\sigma,\tau} \cdot y^{c,v}_{\sigma,\delta}, && \sigma \in \mathcal{A}^{D^{\text{in}}_v}.
\end{align*}
The first inequality models that if there is no pre-image~$\sigma$ under
attack~$\delta$ that is mapped via function~$f^v$ to image~$\tau$ (the
right-hand side of the inequality is~0 in this case), then~$z^{c,v}_{\tau,\delta}$
has to take value~0.
Otherwise, there exists such a pre-image~$\sigma$ and the second
constraints will force~$z^{c,v}_{\tau,\delta}$ to take value~1.

Finally, to ensure that there are no ambiguities, we need to guarantee
that two distinct codewords cannot send the same input to the terminal (due
to an attack).
This can be enforced by adding, for all distinct~$c_1,c_2 \in \mathcal{C}$,
for all attacks~$\delta_1,\delta_2 \in \Delta$, and all
strings~$\tau_1 \in \mathcal{A}^{D^{\text{out}}_1}$ and~$\tau_2 \in
\mathcal{A}^{D^{\text{out}}_2}$, the constraint
\[
  z^{c_1,V_1}_{\tau_1,\delta_1}\cdot z^{c_1,V_2}_{\tau_2,\delta_1}
  +
  z^{c_2,V_1}_{\tau_1,\delta_2}\cdot z^{c_2,V_2}_{\tau_2,\delta_2}
  \leq
  1.
\]
This indeed models nonambiguity, because a summand in this expression can
only take value~1 if under the corresponding attack, the strings~$\tau_1$ and~$\tau_2$
are the results of the functions employed at nodes~$V_1$ and~$V_2$, respectively.

Consequently, any binary assignment to the above variables that satisfies
all constraints encodes an unambiguous outer code.
If one can prove that no such binary assignment exists, then also the
non-existence of unambiguous codes has been shown.
The full MINLP model is given, in condensed form, by the following system.
\begin{align*}
  \min 0, &&&\\
  \sum_{s \in \mathcal{A}} x^c_{a,s} &= 1, &&
  c \in \mathcal{C}, a \in D^{\text{in}}_{V_1} \cup D^{\text{in}}_{V_2},\\
  \sum_{a \in M^v_\delta} x^c_{a,\sigma_a}
  &\leq
    y^{c,v}_{\sigma,\delta} + \lvert M^v_\delta \rvert - 1 &&
  c \in \mathcal{C}, v \in \{V_1,V_2\},
  \sigma \in \mathcal{A}^{D^{\text{in}}_v}, \delta \in \Delta,\\
  y^{c,v}_{\sigma,\delta}
  &\leq
    x^c_{a,\sigma_a}, &&
  c \in \mathcal{C}, v \in \{V_1,V_2\},
  \sigma \in \mathcal{A}^{D^{\text{in}}_v}, \delta \in \Delta,
  a \in M^v_\delta,\\
  \sum_{\tau \in \mathcal{A}^{D^{\text{out}}_v}} f^v_{\sigma,\tau} &= 1,&&
  v \in \{V_1,V_2\},
  \sigma \in \mathcal{A}^{D^{\text{in}}_v},\\
  z^{c,v}_{\tau,\delta}
  &\leq
    \sum_{\sigma \in \mathcal{A}^{D^{\text{in}}_v}} f^v_{\sigma,\tau} \cdot y^{c,v}_{\sigma,\delta}, &&
  c \in \mathcal{C}, v \in \{V_1,V_2\},
  \tau \in \mathcal{A}^{D^{\text{out}}_v},
  \delta \in \Delta,
  \\
  z^{c,v}_{\tau,\delta}
  &\geq
    f^v_{\sigma,\tau} \cdot y^{c,v}_{\sigma,\delta}, &&
  c \in \mathcal{C}, v \in \{V_1,V_2\},
  \sigma \in \mathcal{A}^{D^\text{in}_v}, \tau \in \mathcal{A}^{D^{\text{out}}_v},
  \delta \in \Delta,\\
  1
  &\geq
    z^{c_1,V_1}_{\tau_1,\delta_1}\cdot z^{c_1,V_2}_{\tau_2,\delta_1}
    +
    z^{c_2,V_1}_{\tau_1,\delta_2}\cdot z^{c_2,V_2}_{\tau_2,\delta_2}
    ,&&
  c_1,c_2 \in \mathcal{C}, c_1 \neq c_2,
  \delta_1,\delta_2 \in \Delta,\\
  &&&
  \tau_1 \in \mathcal{A}^{D^{\text{out}}_{V_1}}, \tau_2 \in \mathcal{A}^{D^{\text{out}}_{V_2}},\\
  x^c_{a,s} &\in \{0,1\},&&
  c \in \mathcal{C}, a \in D^{\text{in}}_{V_1} \cup D^{\text{in}}_{V_2},
  s \in \mathcal{A},\\
  f^v_{\sigma,\tau} & \in \{0,1\},&&
  v \in \{V_1,V_2\}, \sigma \in D^{\text{in}}_v, \tau\in D^{\text{out}}_v,\\
  y^{c,v}_{\sigma,\delta} &\in \{0,1\},&&
  c \in \mathcal{C},
  v \in \{V_1,V_2\}, \sigma \in D^{\text{in}}_v, \delta \in \Delta,\\
  z^{c,v}_{\tau,\delta} &\in \{0,1\}, &&
  c \in \mathcal{C},
  v \in \{V_1,V_2\}, \tau \in D^{\text{out}}_v, \delta \in \Delta.
\end{align*}

\newpage

\begin{remark}
  Although the above model contains nonlinear constraints due to the
  products of some variables, it can be linearized as follows.
  For two binary variables~$x$ and~$y$, their product can be expressed
  using an auxiliary variable~$\alpha$ and adding the conditions
  
  $$\alpha
  \leq x,\ \alpha \leq y,\ x + y \leq 1 + \alpha.$$
  To keep the above discussion to the point, we did not incorporate
  this reformulation into the above model.
  Moreover, such a reformulation is automatically carried out by modern
  MINLP technology.
\end{remark}

\begin{remark}
    We mentioned that a similar MINLP model has been used in~\cite{hojny2023role}.
    However, this model did not take adversarial attacks into account.
    To solve both the MINLP model from~\cite{hojny2023role} and the model presented in this section, the state-of-the-art software iteratively creates a list of smaller problems by fixing some variables to~0 or~1, respectively, until a feasible solution is found, or all created subproblems turned out to be infeasible.
    The latter means that there does not exist an unambiguous code.
    We also note that the running time of state-of-the-art software can be exponential in the number of variables.
\end{remark}

\begin{remark}
    The framework of SAT (satisfiability problems) has been used to derive some results presented in Propositon~\ref{prop:familyB}.
    The difference between the framework of MINLP and SAT is that all variables in SAT need to be Boolean, i.e., variables need to take the values ``true'' or ``false''.
    Moreover, the conditions between the variables need to be expressed by Boolean formulas, usually in a disjunctive normal form.
    That is, each condition combines a family of variables by or-operations.

    Since all variables used in our MINLP model are binary variables, there is an immediate transformation to Boolean variables, by interpreting the value~0 as ``false'' and the value~1 as ``true''.
    Furthermore, all constraints can be translated into logical formulas.
    For example, constraint~$y^{c,v}_{\sigma,\delta} \leq x^c_{a,\sigma_a}$ can be written as the logical condition~$x^c_{a,\sigma_a} \vee \neg y^{c,v}_{\sigma,\delta}$.
    For the sake of conciseness, however, we do not provide all constraints in terms of logical relations here.
\end{remark}

\section{Unambiguous Codes for Proving Lower Bounds}
\label{sec:concretecodes}

Proposition~\ref{prop:familyB} and Theorem~\ref{thm:another_interesting_ex} provide lower bounds on the capacity of certain networks.
In this appendix, we provide corresponding outer and inner codes that serve as proof to certify these lower bounds, and briefly describe how we found these codes.

\subsection{Lower Bound in Proposition~\ref{prop:familyB} part 1}

We provide an inner code $\mF= \{\mF_{V_1},\mF_{V_2}\}$ and an outer code $\mC$ of size 10 that is
unambiguous for the channel $\Omega[\mathfrak{B}_2,\mA,\mF,S \to T,\mU_S,1]$ where $\lvert \mA \rvert=4$. Represent $\mA$ as $\{0,1,2,3\}$. We let~$\mF_{V_1}(a)=a$ for all $a \in \mA$ and 
\begin{align*}
\mC &=\{(0,1,2,2),(3,0,1,0),(1,1,3,0),(3,3,2,1),(1,2,0,1),(1,3,1,2),(2,3,0,0),(0,0,0,3),\\&(2,0,3,1),(3,2,3,3)\}.
\end{align*}
Lastly, we let $\mF_{V_2}$:
\begin{itemize}
    \item $\mF_{V_2}^{-1}(0,0) = \{(3,2,1)\}$,
    \item $\mF_{V_2}^{-1}(0,1) = \{(0,0,3), (1,0,3), (2,2,0)\}$,
    \item $\mF_{V_2}^{-1}(0,2) = \{(1,3,0)\}$,
    \item $\mF_{V_2}^{-1}(0,3) = \{(0,0,0), (0,0,2), (0,1,2), (0,1,3), (3,0,2), (3,0,3), (3,1,0), (3,1,3), (3,3,2)\}$,
    \item $\mF_{V_2}^{-1}(1,0) = \{(0,3,1)\}$,
    \item $\mF_{V_2}^{-1}(1,1) = \{(2,3,3)\}$,
    \item $\mF_{V_2}^{-1}(1,2) = \{(1,1,3), (2,0,1)\}$,
    \item $\mF_{V_2}^{-1}(1,3) = \{(0,0,1), (0,2,3), (0,3,3), (1,0,1), (2,0,3), (2,1,3), (2,2,3), (2,3,1)\}$,
    \item $\mF_{V_2}^{-1}(2,0) = \{(1,0,0), (1,0,2), (1,2,0), (1,2,3), (1,3,3), (2,3,0), (2,3,2), (3,3,0), (3,3,3)\}$,
    \item $\mF_{V_2}^{-1}(2,1) = \{(0,2,1), (1,1,2), (2,1,2), (3,1,1), (3,2,2), (3,3,1)\}$,
    \item $\mF_{V_2}^{-1}(2,2) = \{(0,1,0)\}$,
    \item $\mF_{V_2}^{-1}(2,3) = \{(0,2,2), (1,2,2)\}$,
    \item $\mF_{V_2}^{-1}(3,0) = \{(1,2,1), (2,0,0), (2,0,2), (2,1,1), (2,2,1), (2,2,2), (3,0,1), (3,2,0), (3,2,3)\}$,
    \item $\mF_{V_2}^{-1}(3,1) = \{(3,0,0)\}$,
    \item $\mF_{V_2}^{-1}(3,2) = \{(3,1,2)\}$,
    \item $\mF_{V_2}^{-1}(3,3) = \{(0,1,1), (0,2,0), (0,3,0), (0,3,2), (1,1,0), (1,1,1), (1,3,1), (1,3,2), (2,1,0)\}$.
\end{itemize}
It can be verified that the code $\mC$ is unambiguous for the channel
$\Omega[\mathfrak{B}_2,\mA,\mF,S \to T,\mU_S,1]$.

\subsection{Lower Bound in Proposition~\ref{prop:familyB} part 2}

We provide an inner code~$\mF= \{\mF_{V_1},\mF_{V_2}\}$ and an outer code $\mC$ of size 16 that is
unambiguous for the channel $\Omega[\mathfrak{B}_{2},\mA,\mF,S \to T,\mU_S,1]$. Represent $\mA$ as $\{0,1,2,3,4\}$. Let $\mF_{V_1}(a)=a$ for all~$a \in \mA$ and 
\begin{align*}
\mC =\{&(0,0,0,0), (1,4,4,4), (2,0,1,3), (3,4,1,0), (4,2,0,4), (1,2,1,1), (2,4,0,1), (3,3,4,1),\\ & (4,0,2,1), (1,0,3,2), (2,2,4,2), (3,1,2,2), (4,1,4,0), (0,4,2,3), (0,3,3,4), (3,2,3,3)\}.
\end{align*}
Lastly, we let $\mF_{V_2}$:
\begin{itemize}
    \item $\mF_{V_2}^{-1}(0,0) = \{(0,0,0), (3,0,0), (4,0,0), (4,4,0), (2,0,1), (0,1,1), (4,1,1), (2,4,1),$\\
    $\phantom{\mF_{V_2}^{-1}(0,0) = \{}(4,4,1),(0,2,2),(2,3,2),(2,1,3),(0,2,3),(4,0,4),(2,3,4),(2,4,4),$\\
    $\phantom{\mF_{V_2}^{-1}(0,0) = \{}(2,2,0),(3,2,0),(1,3,1),(3,0,2),(3,1,2),(1,0,3),(3,0,3),(1,1,4)\}$,
    \item $\mF_{V_2}^{-1}(0,1) = \{(0,4,4),(4,4,4)\}$,
    \item $\mF_{V_2}^{-1}(0,2) = \{(0,1,3),(1,1,3),(3,1,3),(0,4,3),(0,1,4)\}$,
    \item $\mF_{V_2}^{-1}(0,3) = \{(3,1,0),(4,1,0),(4,3,0),(4,1,2)\}$,
    \item $\mF_{V_2}^{-1}(0,4) = \{(1,0,4),(2,0,4),(2,2,4)\}$,
    \item $\mF_{V_2}^{-1}(1,0) = \{(1,1,1),(2,1,1)\}$,
    \item $\mF_{V_2}^{-1}(1,1) = \{(1,0,1),(4,0,1),(4,3,1),(4,0,2)\}$,
    \item $\mF_{V_2}^{-1}(1,2) = \{(3,4,1),(3,4,3)\}$,
    \item $\mF_{V_2}^{-1}(1,3) = \{(0,2,1),(0,2,4)\}$,
    \item $\mF_{V_2}^{-1}(1,4) = \{(0,3,2),(4,3,2)\}$,
    \item $\mF_{V_2}^{-1}(2,0) = \{(2,4,2)\}$,
    \item $\mF_{V_2}^{-1}(2,1) = \{(1,0,2),(1,1,2),(1,2,2),(3,2,2),(1,2,4)\}$,
    \item $\mF_{V_2}^{-1}(2,2) = \{(1,3,0),(1,4,0),(1,4,3)\}$,
    \item $\mF_{V_2}^{-1}(2,3) = \{(3,2,3),(4,2,3)\}$,
    \item $\mF_{V_2}^{-1}(2,4) = \{(3,3,0),(3,1,4),(3,2,4),(1,3,4),(3,3,4)\}$,
    \item $\mF_{V_2}^{-1}(3,0) = \{(2,3,0),(1,3,3),(2,3,3)\}$,
    \item $\mF_{V_2}^{-1}(3,1) = \{(0,1,0),(2,0,0),(0,0,4),(0,0,3),(4,1,4)\}$,
    \item $\mF_{V_2}^{-1}(3,2) = \{(0,0,1),(0,2,0),(3,0,1),(3,1,1),(2,2,1),(3,2,1),(0,4,1)\}$,
    \item $\mF_{V_2}^{-1}(3,3) = \{(1,4,2),(1,0,0),(0,4,0),(1,2,0),(0,3,0),(0,0,2),(2,2,2),(1,3,2),(0,4,2)\}$,
    \item $\mF_{V_2}^{-1}(3,4) = \{(4,2,1),(1,2,1),(4,2,2),(1,2,3),(4,0,3),(4,4,3),(4,2,4)\}$,
    \item $\mF_{V_2}^{-1}(4,0) = \{(3,4,4),(2,4,0),(3,4,0),(1,4,1),(3,3,1),(3,4,2),(4,4,2),(4,3,4),(1,4,4)\}$,
    \item $\mF_{V_2}^{-1}(4,1) = \{(4,1,3),(1,1,0),(2,1,0),(4,2,0)\}$,
    \item $\mF_{V_2}^{-1}(4,2) = \{(0,3,3),(0,3,1),(0,1,2),(3,3,3)\}$,
    \item $\mF_{V_2}^{-1}(4,3) = \{(3,3,2),(0,3,4),(3,0,4)\}$,
    \item $\mF_{V_2}^{-1}(4,4) = \{(2,3,1),(2,0,2),(2,1,2),(2,0,3),(2,2,3),(4,3,3),(2,4,3),(2,1,4)\}$.
\end{itemize}

\subsection{Lower and Upper Bound in Theorem~\ref{thm:another_interesting_ex} part 1}

The following code has been found by solving the MINLP model of Section~\ref{sec:MINLPmodel} using the software Gurobi version~11.0.3. 
We modeled and solved the MINLP using Gurobi's Python interface.

We provide an inner code~$\mF= \{\mF_{V_1},\mF_{V_2}\}$ and an outer code $\mC$ of size 6 that is
unambiguous for the channel $\Omega[\mathfrak{S}_{3,1,2},\F_2,\mF,S \to T,\mU_S,1]$. Represent $\mA$ as $\{0,1\}$. Let $\mF_{V_1}(a)=a$ for all~$a \in \mA$ and 
\begin{align*}
\mC &=\{(0,0,0,1,1,0),(0,1,1,1,0,1),(0,1,0,0,0,0),(1,0,0,1,0,1),(1,0,1,0,0,0),(1,1,1,1,1,0)\}.
\end{align*}
Lastly, we let $\mF_{V_2}$:
\begin{itemize}
    \item $\mF_{V_2}^{-1}(0,0) = \{(0,0,0), (0,1,1)\}$,
    \item $\mF_{V_2}^{-1}(0,1) = \{(0,0,1), (0,1,0), (1,0,0), (1,1,1)\}$,
    \item $\mF_{V_2}^{-1}(1,0) = \{(1,0,1)\}$,
    \item $\mF_{V_2}^{-1}(1,1) = \{(1,1,0)\}$.
    \end{itemize}

Moreover, to show that no outer code, for the specified parameters, of size~7 exists, we asked Gurobi to solve the MINLP for a code of size~7.
After roughly~7.6 hours of computations, Gurobi concluded that no code of size~7 exists.

\subsection{Lower Bound in Theorem~\ref{thm:another_interesting_ex} part 2}

The following code has been found by solving the MINLP model of Section~\ref{sec:MINLPmodel} using the software Gurobi version~11.0.3. 
We modeled and solved the MINLP using Gurobi's Python interface.
Moreover, we added another constraint to the MINLP model that enforces that the all-0 string is used as codeword, which we can assume without loss of generality.

We provide an inner code~$\mF= \{\mF_{V_1},\mF_{V_2}\}$ and an outer code $\mC$ of size 15 that is
unambiguous for the channel $\Omega[\mathfrak{S}_{3,1,2},\F_3,\mF,S \to T,\mU_S,1]$. Represent $\mA$ as $\{0,1,2\}$. Let $\mF_{V_1}(a)=a$ for all~$a \in \mA$ and 
\begin{align*}
\mC =\{&
(0,0,0,0,0,0),
(2,2,1,0,1,0),
(1,0,0,1,1,0),
(1,1,2,2,1,1),
(2,1,0,2,0,1),\\
&(2,1,2,0,2,2),
(0,2,2,1,1,1),
(1,2,0,1,0,2),
(0,1,0,2,1,2),
(1,1,1,0,0,2),\\
&(0,0,1,2,2,1),
(1,2,2,2,2,2),
(0,1,2,1,0,2),
(2,0,2,1,2,0),
(2,0,1,2,0,0)
\}.
\end{align*}

\newpage

Lastly, we let $\mF_{V_2}$:
\begin{itemize}
    \item $\mF_{V_2}^{-1}(0,0) = \{
    (1,0,2), (2,0,0)\}$,
    \item $\mF_{V_2}^{-1}(0,1) = \{
    (0,0,0), (0,2,2\}$,
    \item $\mF_{V_2}^{-1}(0,2) = \{
    (1,2,0),(2,1,2)\}$,
    \item $\mF_{V_2}^{-1}(1,0) = \{
    (1,2,2),(2,2,1)\}$,
    \item $\mF_{V_2}^{-1}(1,1) = \{
    (0,1,0),(0,1,1),(2,1,1)\}$,
    \item $\mF_{V_2}^{-1}(1,2) = \{
    (0,0,1),(0,2,1),(2,2,0),(2,2,2)\}$,
    \item $\mF_{V_2}^{-1}(2,0) = \{
    (1,1,1),(2,0,1)\}$,
    \item $\mF_{V_2}^{-1}(2,1) = \{
    (0,1,2),(1,1,0),(1,2,1),(2,1,0)\}$,
    \item $\mF_{V_2}^{-1}(2,2) = \{
    (0,0,2),(0,2,0),(1,0,0),(1,0,1),(1,1,2),(2,0,2)\}$.
    \end{itemize}
\end{document}